\documentclass{article}
\usepackage[utf8]{inputenc}
\usepackage{amsmath}
\usepackage{amsfonts}
\usepackage{amssymb}
\usepackage{graphicx}
\usepackage[title,titletoc]{appendix}
\usepackage{fullpage}
\usepackage{color}
\usepackage[sort&compress,numbers]{natbib}
\usepackage{caption}
\usepackage{subcaption}
\usepackage{float}
\usepackage[normalem]{ulem} 
\usepackage{multirow}
\usepackage{caption}

\usepackage{amsthm}
\newtheorem{theorem}{Theorem}[section]
\newtheorem{lemma}[theorem]{Lemma}

\usepackage{xcolor} 

\begin{document}
\title{\normalsize\bf \vspace{-8ex}
Embedding and trajectories of temporal networks}
\author{Chanon Thongprayoon$^1$, Lorenzo Livi$^2$, Naoki Masuda$^{1,3,4,*}$}
\date{\vspace{-5ex}}
\maketitle
\vspace{-14pt}
\begin{center}
\small
$^1$Department of Mathematics, State University of New York at Buffalo, Buffalo, NY, USA.\\
$^2$Department of Computer Science, University of Manitoba, Winnipeg, Manitoba, Canada\\
$^3$Computational and Data-Enabled Science and Engineering Program, State University of New York at Buffalo, Buffalo, NY, USA\\
$^4$Faculty of Science and Engineering, Waseda University, Shinjuku, Tokyo, Japan\\
* Corresponding author: naokimas@buffalo.edu
\end{center}
\text{ {\color{red}This article has been published in IEEE Access, vol. 11, pp. 41426-41443, 2023, doi: 10.1109/ACCESS.2023.3268030.}}

\begin{abstract}
Temporal network data are increasingly available in various domains, and often represent highly complex systems with intricate structural and temporal evolutions.
Due to the difficulty of processing such complex data, it may be useful to coarse grain temporal network data into a numeric trajectory embedded in a low-dimensional space.
We refer to such a procedure as temporal network embedding, which is distinct from procedures that aim at embedding individual nodes.
Temporal network embedding is a challenging task because we often have access only to discrete time-stamped events between node pairs, and, in general, the events occur with irregular intervals, making the construction of the network at a given time a nontrivial question already.
We propose a method to generate trajectories of temporal networks embedded in a low-dimensional space given a sequence of time-stamped events as input. We realize this goal by combining the landmark multidimensional scaling, which is an out-of-sample extension of the well-known multidimensional scaling method, and the framework of tie-decay temporal networks. This combination enables us to obtain a continuous-time trajectory describing the evolution of temporal networks. We then study mathematical properties of the proposed temporal network embedding framework. Finally, we showcase the method with empirical data of social contacts to find temporal organization of contact events and loss of them over a single day and across different days. 
\end{abstract}

\section{Introduction}

Many complex systems can be modeled as networks that consist of nodes and edges connecting nodes \cite{newman2018networks, barabasi2016network}. 
Network analysis has enabled us to extract useful information from data that represent relationships between elements constituting a system and provided concepts and methods with which to understand and intervene into phenomena occurring on networks such as epidemic spreading, cascading failures and synchronization. In fact, the structure of many types of networks changes over time on a timescale that is comparable with the timescale of dynamical processes occurring on networks such as epidemic processes. Examples of such networks, called temporal networks \cite{HOLME201297, holme2015modern, masuda2020guide, bansal2010dynamic}, include social contact networks whose structure evolves owing to mobility of individuals, 
e-mail communication networks among users, and air transportation networks in which direct commercial flights connecting a pair of airports are time-stamped and also vary over months. Research on temporal networks has revealed that the temporality of the networks changes various structural and dynamical properties of networks relative to those of static networks. For example, communities in temporal networks show much richer features than static-network counterparts, such as birth, death, split, and merge of communities, and heterogeneous distributions of lifetimes of individual communities, which are concepts that are relevant only to temporal networks \cite{palla2007quantifying,fortunato2010community}. Such dynamic community structure of temporal networks has been useful for understanding evolution of scientific research fields~\cite{rosvall2010mapping}
and that of groups of senators over the history in the US \cite{mucha2010community}, to name a few.

Temporal networks are intimately related to time series data. However, it seems that temporal network data analyses that have been developed in the past two decades or so have not exploited a full potential of general time series analysis, which has a long history and provides an ample repertoire of methods. Important exceptions include change-point and anomaly detection in temporal networks \cite{chandola2009anomaly, akoglu2015graph,zambon2018concept} and extraction of the discrete state of the temporal network and its transitory dynamics \cite{masuda2019detecting, baker2014fast, ryali2016temporal, VIDAURRE2018646}, which are clearly hinted on the same analyses of general time series data. We propose that this gap is because most methods for analyzing time series assume numeric data, whereas the temporal network data is network-valued by definition. Such methods include the Fourier transform, autocorrelation, curve fitting, fractal analysis, and the Lyapunov exponent, to name a few.
Although one could transform a temporal network data into a numeric time series by, for example, vectorizing the adjacency matrix at each time, such a transformation is usually very high-dimensional because the number of entries of the adjacency matrix scales with $O(N^2)$, where $N$ is the number of nodes. Furthermore, temporal network data are often provided in the form of time-stamped event sequences, and how one should transform such data into numeric vectors in a computationally feasible and informative manner is nontrivial.

To fill this knowledge gap, in the present study, we propose a method for generating a continuous-valued trajectory lying in a Euclidean space given temporal network data as input in the form of time-stamped event sequences; the proposed method is also directly applicable to the input given as a discrete-time sequence of adjacency matrices, which is another major representation of temporal networks \cite{masuda2020guide}.

Network embedding is commonly understood as embedding of nodes of a network into a low-dimensional space by preserving their pairwise similarity as much as possible \cite{cao2015grarep, cavallari2017learning, perozzi2014deepwalk}.
Alternatively, an edge embedding method embeds edges of the given network into a low-dimensional space \cite{grover2016node2vec, wang2020edge2vec}.
An obvious application of network embedding is visualization of networks.
In addition, if network embedding preserves sufficient information of the original network data, one may be able to efficiently carry out downstream data mining tasks by working on the network represented in the embedding space, possibly exploiting machine learning algorithms. Such downstream tasks include link prediction \cite{trouillon2016complex}, visualization \cite{van2008visualizing}, and anomaly detection \cite{agovic2009anomaly}.
There are various node embedding methods for temporal networks, with which nodes move over time in the embedding space \cite{li2017attributed,li2018deep,nguyen2018continuous,goyal2018dyngem,saha2018models}.

In contrast to such methods, the network embedding for temporal networks that we develop in this paper does not belong to the traditional classes of network embedding such as node or edge embedding.
Instead, we aim at embedding temporal networks into a Euclidean space as a single trajectory such that the network at a given point in time is mapped to a point in the embedding space.
The proposed method combines elements of the mathematical modeling framework of tie-decay temporal network \cite{ahmad2021tie} and the landmark multidimensional scaling (LMDS), which is a dimension reduction method for general data \cite{de2004sparse}. The LMDS is an out-of-sample extension of multidimensional scaling (MDS) \cite{france2010two}. This particular combination of methods enables us to generate continuous-time trajectories describing structural and dynamical features of temporal network data given in the form of list of time-stamped events on network edges.

\section{The proposed temporal network embedding algorithm}

\subsection{Tie-decay networks}\label{tie_decay_subsection}
\begin{figure}[t]
    \centering
     \includegraphics[width=0.7\textwidth]{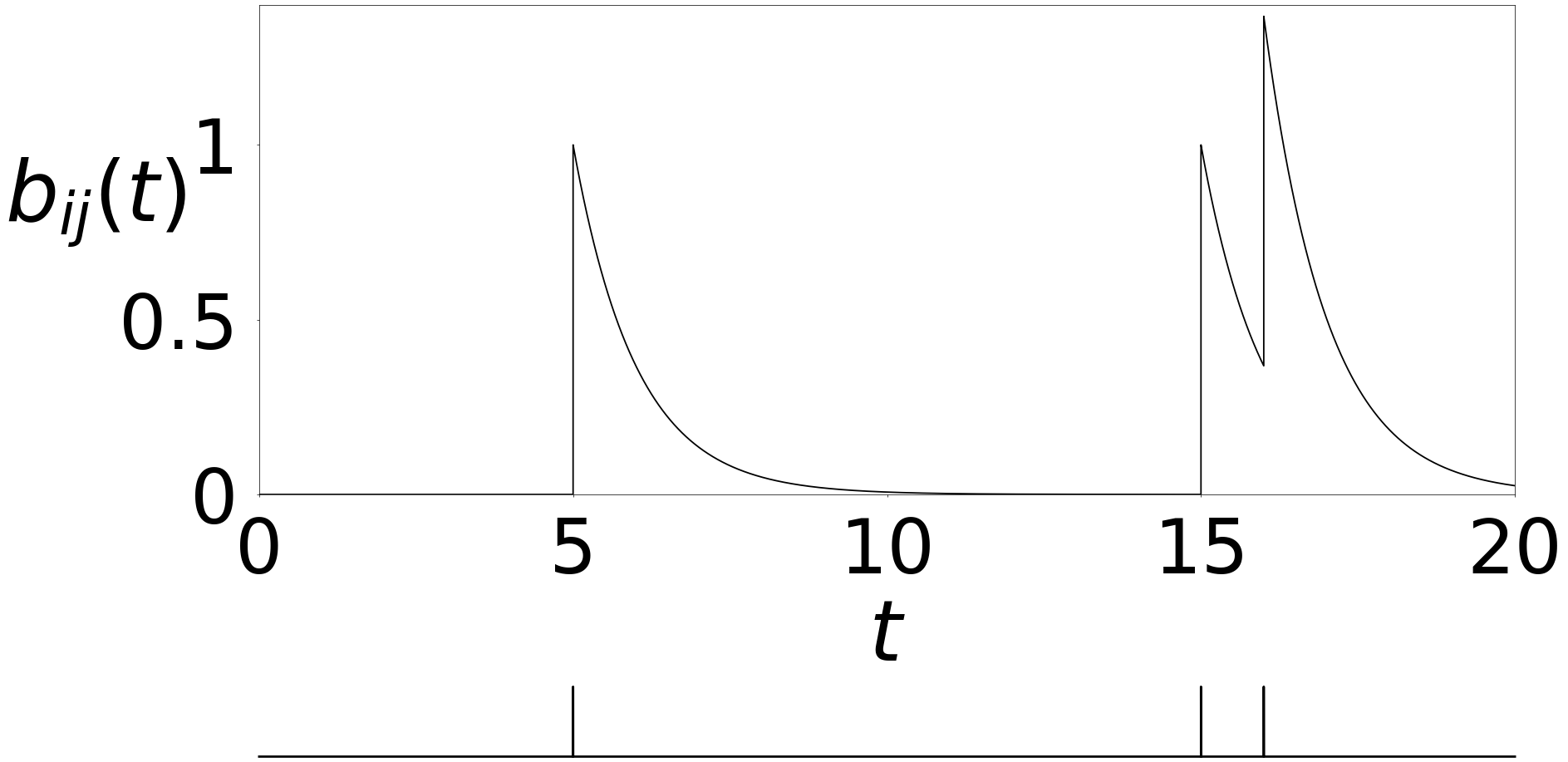}
     \caption{Time course of the edge weight, $b_{ij}(t)$, in a tie-decay network. The vertical bars represent the time of the events, $\tilde{t}_{1} = 5, \tilde{t}_{2} = 15, \tilde{t}_{3}=16$, on edge $(i, j)$. We set $\alpha=1$}
\label{fig:tie_decay_value}
\end{figure}

Consider a set of $N$ nodes and a sequence of time-stamped contact events. A contact event, which we simply refer to as the event in the following text, is assumed to occur between a pair of individuals. Therefore, each event is a tuple $(i_e, j_e, t_e)$, where $i_e$ and $j_e$ are the pair of nodes, and $t_e$ is the time of the event. We allow the event to be directed from $i_e$ to $j_e$.
We seek to embed such a data set into a low-dimensional space as a network trajectory using tie-decay networks.
A tie-decay network is a continuous-time temporal network model in which the weight of the edge (i.e., its tie) increases when an event occurs on the edge, and continuously decays over time~\cite{ahmad2021tie}.
Let $B(t)=[b_{ij}(t)]$, $t\in \mathbb{R}_{\ge 0}$ be a time-dependent $N\times N$ matrix, where $b_{ij}(t)$ represents the weight of directed edge $(i, j)$ at time $t$. We set
\begin{equation}
    b_{ij}(t) = \sum_{\ell} e^{-\alpha(t-\tilde{t}_\ell)}H(t-\tilde{t}_\ell),
\label{eq:b_ij}
\end{equation}
where $H$ is the Heaviside function, and $\tilde{t}_\ell$ is the time of the $\ell$th event on edge $(i, j)$, with
$0\leq \tilde{t}_1 \le \tilde{t}_2 \le \cdots$. Equation~\eqref{eq:b_ij} implies that an event increases the edge weight by 1 and that 
the edge weight exponentially decreases with rate $\alpha$. We show in Fig.~\ref{fig:tie_decay_value} an example time course of $b_{ij}(t)$.

Next, we consider a sequence of time steps $0\leq t_1<t_2<\cdots$, where $t_{\ell}$ is the $\ell$th time when at least one event occurs in the entire network. We denote by $A_{\ell}$ the adjacency matrix of the network constructed from all the events occurring at time $t_\ell$. The input data is then identified with a sequence of adjacency matrices $\{ A_1, A_2, \ldots \}$ and the associated times $\{ t_1, t_2, \ldots \}$. We obtain
\begin{equation}
B(t) = \sum_{\ell=1}^{\infty} e^{-\alpha(t - t_\ell)}A_\ell.
\label{tie_decay_adj_mat_terms}
\end{equation}

\subsection{Classical and landmark multidimensional scaling}\label{MDS_section}

The classical multidimensional scaling embeds a set of high-dimensional data points into a low-dimensional vector space while preserving the dissimilarity between any pair of data points as much as possible \cite{borg2005modern, france2010two, duin2005dissimilarity}. We consider a set of $M$ data points, denoted by $S= \{ x_1, \ldots, x_M \}$ in a metric space $(X,d)$, where $d$ is a distance measure. The MDS embeds each $x_i \in S$ into $\mathbb{R}^k$, where $k$ is relatively small, by carrying out the following steps. First, we construct the squared distance matrix $\Delta = [\Delta_{ij}]$, where $\Delta_{ij} = d^2(x_i,x_j)$. Second, we calculate the double mean-centered dot product matrix $D = -\frac{1}{2}H_M\Delta H_M$, where $H_M = I_M-\frac{1}{M}J_M$, matrix $I_M$ is the $M\times M$ identity matrix, and $J_M$ is the $M\times M$ matrix in which all the elements are equal to 1. Finally, we obtain the embedding coordinates in $\mathbb{R}^k$ for $x_i$ as the $i$th column of the following $k\times M$ matrix:
\begin{equation}
    L_k = \begin{bmatrix}
    \sqrt{\lambda_1}v_1^{\top}\\
    \vdots \\
    \sqrt{\lambda_k}v_k^{\top}
    \end{bmatrix},
    \label{cMDS_coordinates}
\end{equation}
where $\lambda_i$ is the $i$th largest positive eigenvalue of $D$, vector $v_\ell$ is the associated right eigenvector normalized in terms of the $\ell^2$-norm, and ${}^\top$ represents the transposition.

MDS has a time complexity of $O(M^3)$ because it requires eigenvalues and eigenvectors of $D$ \cite{gisbrecht2015metric}. Therefore, it is difficult to apply MDS for large data sets. An out-of-sample extension of a dimension reduction method calculates the exact coordinates in the embedding space only for a small fraction of data points and approximates the coordinates for the remaining data points \cite{Bengio2004Nips, van2009dimensionality}. The LMDS is an out-of-sample extension of the classical MDS \cite{de2004sparse, Desilva2003Nips, Platt2005Aistats}. The LMDS applies the MDS to a subset of the given data points called the landmarks and then uses the landmarks to calculate the embedding coordinates for the other data points. We use $S = \{x_1,\ldots,x_M\}$ as the landmarks and therefore first apply MDS to $S$. The embedding coordinates for $S$ are given by Eq.~\eqref{cMDS_coordinates}. Then, we determine the embedding coordinates of the remaining data as follows.
First, we define a $k\times M$ matrix by
\begin{equation}
    L'_k = \begin{bmatrix}
    \frac{1}{\sqrt{\lambda_1}}v_1^{\top}\\
    \vdots \\ 
    \frac{1}{\sqrt{\lambda_k}}v_k^{\top}
    \end{bmatrix}.
    \label{LMDS_eigen_projection}
\end{equation}
Second, we calculate
 \begin{equation}
        \vec{\delta}_\mu = \frac{1}{M}\sum_{i=1}^{M} \vec{\delta}_i,
        \label{LMDS_avg_coordinates}
\end{equation}
where $\vec{\delta}_i$ is the $i$th column of $\Delta$.
Third, we define the embedding coordinate for an arbitrary out-of-sample data point $z\notin S$ by
\begin{equation}
    \psi(z) = - \frac{1}{2}L'_k(\vec{\delta}_z - \vec{\delta}_\mu),
    \label{LMDS_arbitrary_coordinates}
\end{equation}
where
\begin{equation}
    \vec{\delta}_{z} = \begin{bmatrix}
    d^2(z,x_1)\\ \
    \vdots \\ 
    d^2(z,x_M)
    \end{bmatrix}.
\end{equation}

\subsection{Embedding of tie-decay networks by landmark multidimensional scaling}\label{tie_decay_embedding}

We propose a method to embed tie-decay networks in continuous time into a low-dimensional Euclidean space using LMDS. The key idea is to use the networks at all the times at which the input event happens, i.e., $t \in \{t_1, t_2, \ldots \}$, as landmarks and embed the networks at all the other times, i.e., $t \notin \{t_1, t_2, \ldots \}$, using LDMS. The algorithm proceeds as follows. First, for a given value of $ \alpha$ (i.e., rate of tie decay), we construct the tie-decay network
from the input adjacency matrices, $\{ A_1, \ldots, A_M \}$, using Eq.~\eqref{tie_decay_adj_mat_terms}, where $M$ is the number of the unique time points at which events occur on any edge.
Second, we apply classical MDS to $\{ B(t_1), \ldots, B(t_M) \}$.
Third, we run LDMS with $\{ B(t_1), \ldots, B(t_M) \}$ as landmarks as follows.
We need to embed an arbitrary adjacency matrix $\{B(t) : t\ge 0 \}$, which we call interpolating points when $t\notin \{ t_1, \ldots, t_M \}$, into $\mathbb{R}^k$ to obtain a continuous trajectory.
Using the exponential tie-decay property, we obtain
\begin{equation}
B(t) = e^{-\alpha(t-t_\ell)}B(t_\ell),
\label{matrix_B}
\end{equation}
where $\ell \in \{ 1, \ldots, M\}$ satisfies $t_\ell \le t < t_{\ell+1}$. Then, we use Eq.~\eqref{LMDS_arbitrary_coordinates} to
obtain the embedding coordinates for $B(t)$ for any $t$ as follows:
\begin{equation}
        \psi(B(t)) = -\frac{1}{2}L'_k(\vec{\delta}_B(t)-\vec{\delta}_\mu),
    \label{tie_decay_arbitrary_coordinates}
    \end{equation}
where
\begin{equation}
    \vec{\delta}_B(t) = \begin{bmatrix}
    d^2(B(t),B(t_1))\\
    \vdots \\
    d^2(B(t),B(t_M))
    \end{bmatrix}.
\end{equation}

\subsection{Distance measures for nonnegative matrices}\label{dist_measure}

To calculate $\psi(B(t))$, one needs the squared distance between the matrix $B(t)$ and each of the $M$ landmark matrix representations.
Because $d(B(t),B(t_m))=d(e^{-\alpha (t-t_\ell)}B(t_\ell),B(t_m))$, we should use a distance measure for networks for which we can easily calculate $d(c M_1, M_2)$
given $d(M_1, M_2)$, where $M_1$ and $M_2$ are general weighted adjacency matrices and $0 < c < 1$.
Here we provide three eligible distance measures.

\subsubsection{Frobenius distance}

Let $\langle\cdot,\cdot\rangle$ be the Frobenius inner product for real-valued matrices given by
\begin{equation}
\langle P,Q\rangle = \sum_{i = 1}^a\sum_{j = 1}^b p_{ij}q_{ij}\,,
\end{equation}
where $P = [p_{ij}]$ and $Q = [q_{ij}]$ are $a\times b$ matrices.
We define the induced Frobenius norm $\|\cdot\|$ by
\begin{equation}
\| P\| = \sqrt{\langle P,P\rangle}\ .
\end{equation}
We obtain the Frobenius distance by $d_F(B(t_\ell),B(t_{\ell'}))=\|B(t_\ell)-B(t_{\ell'})\|$. For $\ell,\ell'\in\{1, \ldots, M \}$, we obtain
\begin{equation}
    d^2_F(B(t_\ell),B(t_{\ell'})) = \|B(t_\ell)\|^2 - 2\langle B(t_\ell),B(t_{\ell'})\rangle + \|B(t_{\ell'})\|^2.
\label{eq:d^2_F(B_ell,B_m)}    
\end{equation}
Using Eq.~\eqref{eq:d^2_F(B_ell,B_m)}, we obtain for any $t\in [t_\ell, t_{\ell+1})$
\begin{equation}
    d_F(e^{-\alpha(t-t_\ell)}B(t_\ell),B(t_{\ell'}))=\sqrt{e^{-\alpha(t-t_\ell)}(e^{-\alpha(t-t_\ell)}-1)\|B(t_\ell)\|^2 + (1-e^{-\alpha(t-t_\ell)})\|B(t_{\ell'})\|^2 + e^{-\alpha(t-t_\ell)}d^2(B(t_\ell),B(t_{\ell'}))}\,. \label{euclideandist}
\end{equation}
Therefore, we can easily compute $\psi(B(t))$ for any $t$ by precomputing $\|B(t_\ell)\|^2$ for $\ell \in \{1, \ldots, M\}$ in addition to matrix $\Delta$.

\subsubsection{Laplacian distance}

In this and next subsections, we assume undirected networks. The Laplacian matrix $L(t)$ of a tie-decay network at time $t$ is defined by $L(t)\equiv D^{\ast}(t) - B(t)$, where $D^{\ast}(t)$ is the diagonal matrix whose $i$th diagonal entry is given by $\displaystyle\sum_{j=1}^N b_{ij}(t)$. The unnormalized Laplacian distance between two Laplacian matrices, which we refer to as the Laplacian distance for short in the following text, is defined by
\begin{equation}
    d_L(L(t_\ell),L(t_{\ell'})) = \sqrt{\sum_{k=1}^N \left[\rho_k(L(t_\ell)) - \rho_k(L(t_{\ell'})) \right]^2}\,,
\label{eq:unnormalized-Laplacian-distance-def}
\end{equation}
where  $\rho_k(P)$ is the $k$th smallest eigenvalue of symmetric matrix $P$~\cite{donnat2018tracking,wilson2008study, masuda2019detecting}. It should be noted that one can replace $\sum_{k=1}^N$ in Eq.~\eqref{eq:unnormalized-Laplacian-distance-def} by $\sum_{k=2}^{N}$ because the smallest eigenvalue of a Laplacian matrix is equal to 0. For any $t$ satisfying $t_{\ell} \le t < t_{\ell+1}$, we obtain $L(t) = e^{-\alpha(t-t_\ell)}L(t_\ell)$ and therefore
\begin{align}
    d_L(e^{-\alpha(t-t_\ell)}L(t_\ell), L(t_{\ell'})) 
    =& \sqrt{\sum_{k=2}^N \left[e^{-\alpha(t-t_\ell)}\rho_k(L(t_\ell)) - \rho_k(L(t_{\ell'})) \right]^2}.
\end{align}

\subsubsection{Bregman distance}

For a pair of networks $G_1$ and $G_2$ each with $N$ nodes, we consider the distance measure given by
\begin{equation}
F(G_1,G_2) = \sum_{k = 1}^{N}\sum_{k' = 1}^{N} f(\rho_{k}(L_1),\rho_{k'}(L_2))\left| \langle u_{k}(L_1),u_{k'}(L_2)\rangle_E\right|^m,
\label{eq:general_distance_laplacian}
\end{equation}
where $f:\mathbb{R}^2\rightarrow\mathbb{R}_{\ge 0}$ is symmetric with respect to the two arguments and satisfies the triangle inequality; $L_i$ (with $i\in \{1, 2 \}$) is the Laplacian matrix of $G_i$; $u_{k}(L_i)$ is an eigenvector of $L_i$ that is normalized in terms of the $\ell^2$-norm and associated with eigenvalue $\rho_{k}(L_i)$; $\langle \cdot,\cdot\rangle_E$ is the Euclidean inner product of two vectors; $m$ is a positive integer~\cite{wicker2013new}. Using the tie-decay networks at times $t$ and $t_m$ as the input to $F$, where $t_{\ell} \le t < t_{\ell+1}$, we obtain
\begin{equation}
F(B(t), B(t_{\ell'})) = \sum_{k = 1}^{N}\sum_{k' = 1}^{N} f(e^{-\alpha(t-t_{\ell})}\rho_{k}(L(t_{\ell})), \rho_{k'}(L(t_{\ell'})))\left| \langle u_{k}(L(t_{\ell}),u_{k'}(L(t_{\ell'}))\rangle_E\right|^m\,.
\label{eq:general_distance_laplacian_2}
\end{equation}

We assume $m = 2$, which is known to make Eq.~\eqref{eq:general_distance_laplacian} independent of the choice of the eigenvectors~\cite{wicker2013new}. An example of $F$ is given through the Bregman divergence as follows. Let $g: \Omega\rightarrow\mathbb{R}$ be a continuously differentiable and strictly convex function, where $\Omega\subseteq\mathbb{R}^n$. We define the Bregman divergence associated with $g$, which we denote by $d_g$, by
\begin{equation}
d_g(x,y) = g(x) - g(y) - \langle x-y,\nabla g(y)\rangle_E\,,
\label{Bregman_div}
\end{equation}
where $x, y \in \mathbb{R}^n$, and $\nabla g(y)$ denotes the gradient of $g(y)$. If $g(x) = \langle x,x\rangle_E$, then $d_g$ reduces to the squared $\ell^2$-norm because
\begin{align}
d_{\langle\cdot,\cdot\rangle_E}(x,y) &= \langle x,x\rangle_E - \langle y,y\rangle_E - \langle x-y,\nabla\langle y,y\rangle_E\rangle_E\notag\\
&= \langle x,x\rangle_E - \langle y,y\rangle_E - \langle x-y,2y\rangle_E\notag\\
&= \left| x - y\right|_2^2,
\end{align}
where $\left|\cdot\right|_2$ denotes the $\ell^2$-norm~\cite{banerjee2005clustering}. By setting $f = d_{\langle\cdot,\cdot\rangle_E}$, $n=1$, and $m=2$ in Eq.~\eqref{eq:general_distance_laplacian_2}, we obtain the Bregman distance, denoted by $d_B$, given by
\begin{equation}
    d_B(B(t), B(t_{\ell'})) = \sum_{k= 1}^N\sum_{k' = 1}^N
    \left[e^{-\alpha(t-t_{\ell})} \rho_{k}(L(t_\ell)) - \rho_{k'}(L(t_{\ell'}))\right]^2\langle u_{k}(L(t_\ell)), u_{k'}(L(t_{\ell'})\rangle_E^2\,,
\label{eq:d_B-def}
\end{equation}
where $t_\ell \le t < t_{\ell+1}$.
A modification of the Bregman distance that better captures some global dissimilarities between two networks such as connectivity is given by \cite{wicker2013new} 
\begin{equation}
    d_{B'}(B(t),B(t_{\ell'})) = \sum_{k = 1}^N\sum_{k' = 1}^N \frac{\left[e^{-\alpha(t-t_{\ell})}\rho_{k}(L(t_\ell))-\rho_{k'}(L(t_{\ell'}))\right]^2}
    {e^{-\alpha(t-t_{\ell})}\rho_{k}(L(t_\ell))+\rho_{k'}(L(t_{\ell'}))}\langle u_k(L(t_\ell)), u_{k'}(L(t_{\ell'}))\rangle_E^2\,.
\label{eq:d_B'-def}
\end{equation}

\subsection{Goodness of temporal network embedding}\label{goodness_of_fit_criteria}

MDS better preserves the pairwise dissimilarity between data points if the embedding dimension, $k$, is larger. Furthermore, the quality of the embedding produced by MDS may depend on the distance measure.
To select a distance measure and $k$ for MDS, we use the following goodness of fit criteria.

The first criterion uses the fact that embedding into the Euclidean space preserves the pairwise dissimilarity for all pairs of data points if and only if matrix $D = -\frac{1}{2}H_M\Delta H_M$ is positive semidefinite \cite{graepel1999classification, pekalska2001generalized}. 
The existence and magnitude of negative eigenvalues indicate the deviation from the Euclidean behavior.
Therefore, as goodness of fit, we use the ratio of the $k$ most positive eigenvalues over the sum of the absolute values of all eigenvalues \cite{graffelman2020goodness}.
Specifically, we assume that $D$ has at least $k$ positive eigenvalues and define
\begin{equation}
g = \dfrac{\sum_{\ell = 1}^k \lambda_\ell}{\sum_{\ell' = 1}^M \left| \lambda_{\ell'}\right|}\,,
\label{eq:eigenvalue_ratio}
\end{equation}
where $\lambda_1 > \cdots > \lambda_M$ are the eigenvalues of $D$. Index $g$ ranges between $0$ and $1$ and it is close to 1 when the embedding is quasi-Euclidean.

Another goodness of fit criterion for MDS is given by the normalized stress function~\cite{kruskal1964multidimensional,borg2005modern,duin2005dissimilarity}, denoted by $\sigma$, which is defined as
\begin{equation}
\sigma = \sqrt{\frac{\sum_{\ell=1}^M \sum_{\ell' = 1}^{\ell - 1} (d(x_{\ell'},x_{\ell}) - \left|\psi(x_{\ell'}) - \psi(x_{\ell}) \right|_2)^2}{\sum_{\ell=1}^M \sum_{\ell' = 1}^{\ell-1} d^2(x_{\ell'},x_\ell)}}\,.
\label{eq:stress_function}
\end{equation}
Note that
$\left|\psi(x_{\ell'}) - \psi(x_{\ell}) \right|_2$ is the distance between the two data points in the embedding space. A small $\sigma$ value indicates that MDS is accurate at representing the dissimilarity structure of the original data set.
Values of $\sigma$ less than $0.15$ have been suggested to be acceptable~\cite{borg2005modern}.

\subsection{Time complexity of the temporal network embedding}

Given a squared distance matrix $\Delta$, the computation of the embedding coordinate, given by~\eqref{tie_decay_arbitrary_coordinates}, consists of 
computing the double mean-centered dot product matrix $D$, which requires $O(M^3)$ time, 
computing the eigenvalues and eigenvectors of $D$, which requires $O(M^3)$ time, and
then computing $L'_k$, which requires $O(kM)$ time,
computing $\delta_\mu$, which requires $O(M^2)$ time, computing $\delta_B(t) - \delta_\mu$, which requires $O(M)$ time,
and computing $L'_k\left(\delta_B(t) - \delta_\mu\right)$, which requires  $O(kM)$ time. The total time complexity of these computations combined is $O(M^3)$. The time complexity for calculating $\delta_B(t)$ depends on the choice of the network distance measure. The time complexity for embedding $B(t)$ for any $t\in [t_{\tilde{\ell}}, t_{\tilde{\ell}+1})$ is as follows.

In the case of the Frobenius distance, the time complexity for computing the distance between two landmarks is $O(N^2)$. Therefore, constructing the $M\times M$ squared distance matrix takes $O(M^2N^2)$ time. By precalculating $\left\| B(t_\ell)\right\| \forall\ell\in\{1,\ldots, M\}$, which takes $O(MN^2)$ time, and storing it and the pairwise distance between the landmarks, we can compute $\delta_B(t)$ in $O(M)$ time. It follows that the total time complexity for embedding $B(t)$ is $O(M^2N^2 + M^3)$.

In the case of the Laplacian distance, the time complexity for computing the Laplacian matrices of all landmarks is $O(MN^2)$. To compute the distance between two landmarks, we need the eigenvalues of $L(t_\ell) \forall\ell\in\{1,\ldots, M\}$. Because the time complexity for calculating the eigenvalues is $O(N^3)$, it takes $O(MN^3)$ to compute $\rho_k\left(L(t_\ell)\right)$ $\forall k\in\{2,\ldots, N\}, \forall \ell\in\{1,\ldots, M\}$. Once
we store $\rho_k\left(L(t_\ell)\right)$ $\forall k\in\{2,\ldots, N\}, \forall \ell\in\{1,\ldots, M\}$, it takes $O(M^2N)$ time to compute the squared distance matrix using~\eqref{eq:unnormalized-Laplacian-distance-def} and $O(MN)$ time to compute $\delta_B(t)$. Therefore, the total time complexity for embedding $B(t)$ is $O(MN^3 + M^2N + M^3)$.

Finally, we consider the case of the Bregman distance. In addition to the time complexity of $O(MN^2)$ for computing the Laplacian matrices of all landmarks, it takes 
$O(MN^3)$ time to compute all the eigenvalues and eigenvectors of all the Laplacian matrices, i.e., $\rho_{k}(L(t_\ell))$ and $u_k(L(t_\ell))$, $\forall k$, $\forall t_{\ell}$. Given them, it takes $O(M^2N^3)$ time to compute the squared Bregman distance matrix using~\eqref{eq:d_B-def} or the squared modified Bregman distance using~\eqref{eq:d_B'-def}. Then, the time complexity for computing $\delta_B(t)$ is $O(MN^3)$. Therefore, the total time complexity of embedding $B(t)$ is $O(M^2N^3 + M^3)$.

\section{Theoretical results}

In this section, we mathematically investigate two properties of the proposed embedding algorithm for temporal networks: periodicity and boundedness.

\subsection{Periodic input events}\label{periodic_input}

We first examine the convergence of the trajectory in the embedding space when the input events periodically arrive over time. Let $\{A_\ell\}_{\ell=1}^\infty$ be an infinite sequence of adjacency matrices with each $A_\ell$ encoding the new network structure at time $t_\ell$. We assume that $\{A_\ell\}_{\ell=1}^\infty$ is periodic such that $A_{qT +k} = \tilde{A}_k$ and $t_{qT+k} = q\tilde{T} + t_k$ for each $q\in\mathbb{N}\cup\{0\}$ and $k\in\{1,\ldots, T\}$, where $T\in\mathbb{N}$ is the number of time points in one cycle at which at least one event occurs, and $\tilde{T}$ is the period in terms of absolute time.
\begin{lemma}\label{limitcycle_lemma}
For any $t\in [0, \infty)$, write $t = q\tilde{T} + t'$ such that $t_k \leq t' < t_{k+1}$ for unique values of $q\in\mathbb{N}\cup\{0\}$ and $k \in \{1, \ldots, T\}$. 
In the case of $k = T$, we interpret $t_{k+1} = \tilde{T} + t_1$. Then, we obtain $B(t) = \displaystyle\sum_{\ell=1}^{T} c_\ell(t)A_\ell$, where 
\begin{equation}
    c_\ell(t) = \begin{cases}
    \dfrac{1-e^{-\alpha(q+1)\tilde{T}}}{1-e^{-\alpha\tilde{T}}}e^{-\alpha(t-q\tilde{T}-t_{\ell})}\text{ if }\ 1\leq \ell\leq k,\\
    \dfrac{1-e^{-\alpha q\tilde{T}}}{1-e^{-\alpha\tilde{T}}}e^{-\alpha(t-(q-1)\tilde{T}-t_{\ell})}\text{ if }\ k+1\leq \ell\leq T.
    \end{cases}
    \label{periodic_coeff}
\end{equation}
In particular, the tie-decay network with the periodic input converges to the limit cycle of period $\tilde{T}$ given by
\begin{equation}
X(t') \equiv \sum_{\ell=1}^{T} \frac{e^{-\alpha(t' - t_{\ell})}}{1 - e^{-\alpha\tilde{T}}} A_{\ell}\,,
\label{eq:def limit cycle}
\end{equation}
where $0\le t'<\tilde{T}$.
\end{lemma}
\begin{proof}
Using Eq.~\eqref{tie_decay_adj_mat_terms}, we obtain
\begin{align}
B(t) =& \sum_{\ell=1}^{qT + k} e^{-\alpha (t - t_{\ell})}A_{\ell}\notag\\
  =& \sum_{\ell=1}^{k}\sum_{j=0}^{q} e^{-\alpha\left[t - (j\tilde{T} + t_{\ell})\right]}A_\ell + \sum_{\ell = k+1}^{T}\sum_{j = 0}^{q-1}e^{-\alpha \left[t - (j\tilde{T} + t_{\ell})\right]}A_\ell\notag\\
  =& \sum_{\ell=1}^{k}e^{-\alpha(t'-t_{\ell})}A_\ell\sum_{j=0}^{q} e^{-\alpha (q-j)\tilde{T}} + \sum_{\ell = k+1}^{T}e^{-\alpha(\tilde{T}+t'-t_{\ell})} A_\ell\sum_{j = 0}^{q-1}e^{-\alpha(q-1-j)\tilde{T}}\notag\\
  =&\sum_{\ell=1}^{k}\frac{1-e^{-\alpha(q+1)\tilde{T}}}{1-e^{-\alpha\tilde{T}}}e^{-\alpha(t'-t_{\ell})}A_\ell + \sum_{\ell=k+1}^{T}\frac{1-e^{-\alpha q\tilde{T}}}{1-e^{-\alpha\tilde{T}}}e^{-\alpha(\tilde{T}+t'-t_{\ell})}A_\ell\notag\\
=& \sum_{\ell=1}^{T} c_{\ell}(t)A_{\ell}.
\label{eq:derive c_ell}
\end{align}
Equation~\eqref{periodic_coeff} implies that 
\begin{equation}
    \lim_{q\rightarrow\infty} c_{\ell}(q\tilde{T} + t') = \frac{e^{-\alpha(t' - t_{\ell})}}{1 - e^{-\alpha\tilde{T}}}
\end{equation}
for any $\ell\in\{1,\ldots, T\}$ and any $t'\in [0, \tilde{T})$. Therefore, we obtain Eq.~\eqref{eq:def limit cycle}.
\end{proof}

\begin{theorem}\label{thm:periodic}
The embedding coordinates of tie-decay networks, given by Eq.~\eqref{tie_decay_arbitrary_coordinates}, with periodic input converge to a limit cycle.
\end{theorem}
\begin{proof}
Under the periodic input, Lemma~\ref{limitcycle_lemma} guarantees that $\displaystyle\lim_{q\to\infty} B(q\tilde{T} + t')$ converges to the limit cycle $X(t')$. 
Furthermore, because of the continuity of the distance metric function, $d$, the $\displaystyle\lim_{q\to\infty}$ operation and $d$ are interchangeable. Therefore, we obtain
\begin{equation}
\displaystyle\lim_{q\to\infty}d(B(q\tilde{T} + t'), B(t_\ell)) = d(\displaystyle\lim_{q\to\infty}B(q\tilde{T} + t'), B(t_\ell)) = d(X(t'), B(t_\ell))
\label{eq:lim-and-d-exchange}
\end{equation}
for any $\ell \in \{ 1, \ldots, T\}$. Using Eq.~\eqref{eq:lim-and-d-exchange}, we obtain
\begin{equation}
    \lim_{q\to\infty} \Vec{\delta}_{B}(q\tilde{T} + t') = \lim_{q\to\infty}\begin{bmatrix} d(B(q\tilde{T} + t'), B(t_1))\\
    \vdots\\
    d(B(q\tilde{T} + t'), B(t_T))
    \end{bmatrix} 
%
%
= \Vec{\delta}_X(t').
\label{eq:landmark_dist_limit_cycle}
\end{equation}
By combining Eqs.~\eqref{tie_decay_arbitrary_coordinates} and~\eqref{eq:landmark_dist_limit_cycle}, we obtain
\begin{align}
     \lim_{q\to\infty}\psi(B(q\tilde{T} + t'))
=  - \lim_{q\to\infty} \frac{1}{2}L'(\vec{\delta}_{B}(q\tilde{T} + t')-\vec{\delta}_\mu).
    =&-\frac{1}{2}L'(\Vec{\delta}_X(t')-\vec{\delta}_\mu) = \psi(X(t')).
\end{align}
Therefore, the limit cycle in the embedding space is given by $\psi(X(t'))$, where $0\le t' < \tilde{T}$.
\end{proof}

To illustrate Theorem~\ref{thm:periodic} by an example, we consider a temporal network on $N=3$ nodes. We give a periodic input with $T=4$ input matrices defined by
\begin{equation}
\tilde{A}_{1} = \begin{bmatrix}
0 & 1 & 0\\
1 & 0 & 1\\
0 & 1 & 0
\end{bmatrix}, \quad
\tilde{A}_{2} = \begin{bmatrix}
0 & 1 & 0\\
1 & 0 & 0\\
0 & 0 &0
\end{bmatrix}, \quad
\tilde{A}_{3} = \begin{bmatrix}
0 & 0 & 1\\
0 & 0 & 1\\
1 & 1 & 0
\end{bmatrix}, \quad
\tilde{A}_{4} = \begin{bmatrix}
0 & 1 & 1\\
1 & 0 & 1\\
1 & 1 & 0
\end{bmatrix}.
\end{equation}
We set $t_{\ell} = \ell$ for each $\ell\in\mathbb{N}$ and $\alpha = 0.15$. The network trajectory in the two-dimensional embedding space is shown in Fig.~\ref{fig:limit_cycle}. Figures~\ref{fig:limit_cycle}(a) and \ref{fig:limit_cycle}(b) show a trajectory obtained with the Frobenius and the Laplacian distance measure, respectively. As expected, both trajectories approach a limit cycle with period $4$.

Under the Frobenius distance measure and $\alpha = 0.15$, we obtain $g = 0.617$
%
%
and $\sigma = 0.32$ with $k=1$, $g = 0.951$
%
%
and $\sigma = 0.06$ with $k=2$, and $g > 0.999$
%
%
and $\sigma = 9.35\times 10^{-15}$ with $k=3$. We remind that $g$ and $\sigma$ are goodness of embedding measures defined by Eqs.~\eqref{eq:eigenvalue_ratio} and \eqref{eq:stress_function}, respectively. Under the Laplacian distance measure,
$D$ has only two positive eigenvalues. Therefore, we consider $g$ and $\sigma$ when $k =1,2$. 
When $\alpha = 0.15$, we obtain $g = 0.820$
%
%
and $\sigma = 0.26$ with $k=1$ and $g > 0.999$
%
%
and $\sigma = 5.34\times 10^{-15}$ with $k=2$. Therefore, for both distance measures, we conclude that $k=2$ is sufficient for embedding this tie-decay network.

\begin{figure}[t]
    \centering
     \includegraphics[width=\textwidth]{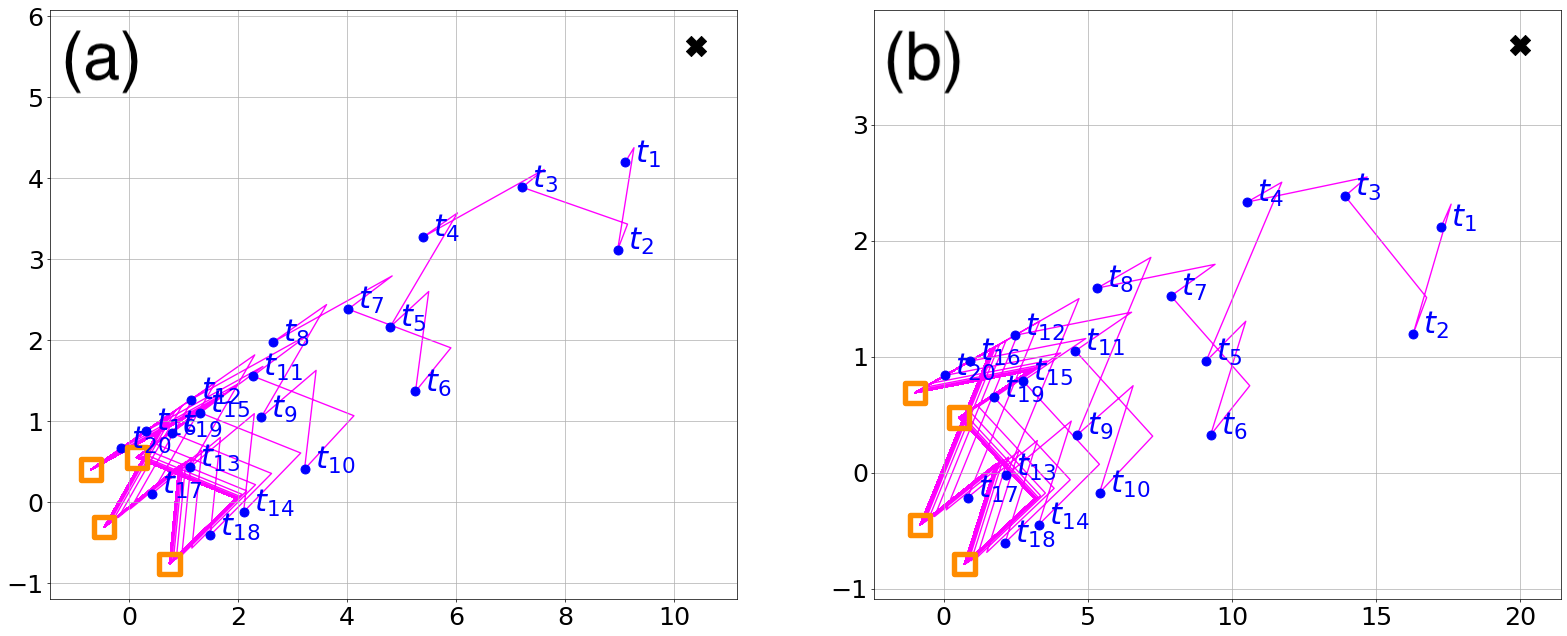}
     \caption{Network trajectories with periodic input. (a) A trajectory under the Frobenius distance measure. (b) A trajectory under the Laplacian distance measure. The input adjacency matrices are periodic with four adjacency matrices. We set $\alpha = 0.15$. The squares indicate the coordinates of the limit cycle right after an input adjacency matrix arrives. The crosses indicate the coordinate of the empty network.}
\label{fig:limit_cycle}
\end{figure}

\subsection{Boundedness of trajectories}\label{subsect:bound_coeff}

Next we investigate boundedness of the network trajectory, i.e., the property of the trajectory that it remains within a particular bounded domain. We assume that adjacency matrices $\{A_\ell \}_{\ell=1}^{\infty}$ arrive with regular intervals, i.e., $t_\ell = \ell\tau$, where $\tau (>0)$ is a constant. However, the extension of the following results to the case of aperiodic inputs, with the minimal interval between two consecutive input adjacency matrices being equal to $\tau$, is straightforward. We also assume that each adjacency matrix $A_\ell$ is selected from a finite set, which we denote by $\{ \tilde{A}_m \}_{m=1}^M$. In the periodic input case discussed in the previous section, we obtain $M=T$ and $A_\ell = \tilde{A}_{((\ell -1) \text{ mod } T)+1}$.
\begin{lemma}\label{lemma:bound_coeff}
For any sequence of input adjacency matrices $\{A_\ell\}_{\ell=1}^\infty$, where each $A_\ell\in\{ \tilde{A}_1,\ldots,\tilde{A}_M\}$, let $B(t)$ be the corresponding tie-decay matrix, which we write as 
\begin{equation}
B(t) = \displaystyle\sum_{m=1}^{M}\tilde{c}_m(t)\tilde{A}_m
\label{eq:B(t)_for_boundedness}
\end{equation}
with $\tilde{c}_m(t) \ge 0$. Then, each $\tilde{c}_m(t)$ is bounded.
\end{lemma}
\begin{proof}
We first show that $\displaystyle\sum_{m=1}^{M}\tilde{c}_m(t)$ is bounded for any input sequence, hence so is each $\tilde{c}_m(t)$. Consider an arbitrary $t\in[t_{\ell},t_{\ell+1})$, for which we can write $t = t_\ell + \Delta t$ with $\Delta t\in [0,\tau)$. For each $m\in\{1,\ldots,M\}$, we obtain
\begin{align}
    0\leq \tilde{c}_{m}(t) &\leq \sum_{m'=1}^{M} \tilde{c}_{m'}(t)\notag\\
    &= \sum_{j=1}^{\ell} e^{-\alpha(t - t_j)}\notag\\
    &= e^{-\alpha \Delta t} \sum_{j = 0}^{\ell-1} e^{-\alpha j\tau}\notag \\
    &\leq e^{-\alpha \Delta t} \frac{1-e^{-\alpha \ell \tau}}{1-e^{-\alpha\tau}}\notag\\
    &< \frac{1}{1 - e^{-\alpha\tau}}\,.
    \label{eq:boundedness_for_c_ell}
\end{align}
\end{proof}

Now we show a boundedness of the network trajectory constructed with the Frobenius distance.
\begin{theorem}
If we use the Frobenius distance for MDS, the embedding coordinate $\psi(B(t))$ for any $t \in [0, \infty)$ is bounded as follows:
\begin{equation}
\left|\psi(B(t))\right|_2 < \frac{1}{2}\left[\sqrt{M}\left(\frac{\sum_{m=1}^M\left\|\tilde{A}_m\right\|}{1 - e^{-\alpha\tau}}\right)^2 \sqrt{\sum_{\ell = 1}^k \frac{1}{\lambda_\ell}} + \left|L'_k\vec{\delta}_\mu \right|_2\right].
\label{eq:B(t)-bounded-final}
\end{equation}
\end{theorem}
\begin{proof}

Using Eqs.~\eqref{eq:B(t)_for_boundedness} and \eqref{eq:boundedness_for_c_ell}, we obtain
\begin{align}
    d_F(B(t),B(t_\ell)) &= \left\|B(t) - B(t_\ell)\right\|\notag\\
    &= \left\|\sum_{m=1}^M \left[\tilde{c}_m(t)-\tilde{c}_m(t_\ell)\right]\tilde{A}_m\right\|\notag\\
    &\le \sum_{m=1}^M \left| \tilde{c}_m(t) -\tilde{c}_m(t_\ell)\right|\left\|\tilde{A}_m\right\|\notag\\
    &< \frac{1}{1 - e^{-\alpha\tau}}\sum_{m=1}^M\left\|\tilde{A}_m\right\|.
\label{eq:Frob_measure_bound}
\end{align}
Next, we evaluate the distance of the embedding coordinates of $B(t)$ from the origin. Using Eq.~\eqref{tie_decay_arbitrary_coordinates}, we obtain
\begin{align}
\left|\psi(B(t))\right|_2&= \left|-\frac{1}{2}L'_k(\vec{\delta}_B(t)-\vec{\delta}_\mu)\right|_2\notag\\
&\leq \frac{1}{2}\left(\left|L'_k \vec{\delta}_B(t)\right|_2 + \left|L'_k\vec{\delta}_\mu\right|_2 \right).
\label{eq:bound_frob_coord_1}
\end{align}
We obtain
\begin{align}
\left|L'_k\vec{\delta}_B(t)\right|_2^2
    &= \sum_{\ell = 1}^k \langle \frac{1}{\sqrt{\lambda_{\ell}}} v_{\ell},\vec{\delta}_B(t) \rangle_E^2.\label{eq:L'delta}
\end{align}
By applying Cauchy-Schwarz inequality to Eq.~\eqref{eq:L'delta} and using the fact that the eigenvectors of $D$ are normalized, we obtain
\begin{align}
\left|L'_k\vec{\delta}_B(t)\right|_2^2 &\le \sum_{\ell=1}^k \frac{1}{\lambda_{\ell}}\left| v_{\ell}\right|_2^2 \left|\vec{\delta}_B(t)\right|_2^2\notag\\
&= \sum_{\ell=1}^k \frac{1}{\lambda_{\ell}} \left|\vec{\delta}_B(t)\right|_2^2\notag\\
&= \sum_{\ell=1}^k \frac{1}{\lambda_{\ell}}\sum_{\ell' = 1}^M d_F^4(B(t),B(t_{\ell'}))\label{eq:L'delta_distance}
\end{align}
By combining Eqs.~\eqref{eq:Frob_measure_bound} and \eqref{eq:L'delta_distance}, we obtain
\begin{align}
\left|L'_k\vec{\delta}_B(t)\right|_2^2 &< \sum_{\ell=1}^k \frac{1}{\lambda_{\ell}}\sum_{\ell'=1}^M \left[\frac{1}{\left(1-e^{-\alpha\tau}\right)^4}\sum_{m=1}^M \left(\left\| \tilde{A}_m\right\|\right)^4\right]\notag\\
&= \sum_{\ell=1}^k \frac{1}{\lambda_{\ell}} \left[\frac{M}{\left(1-e^{-\alpha\tau}\right)^4}\sum_{m=1}^M \left(\left\| \tilde{A}_m\right\|\right)^4\right]. \label{eq:L'delta_bound}
\end{align}
By substituting Eq.~\eqref{eq:L'delta_bound} in Eq.~\eqref{eq:bound_frob_coord_1}, we obtain Eq.~\eqref{eq:B(t)-bounded-final}.

\end{proof}

Next we show the boundedness when we use the Laplacian distance for the MDS. Recall that $L(t)$ denotes the Laplacian matrix of the tie-decay matrix $B(t)$.
\begin{lemma}\label{lemma_bound}
Let $M$ be an $N\times N$ matrix, and $L_M$ be the Laplacian matrix induced by $M$. Then, the map $M\mapsto L_M$ is a linear map. 
\end{lemma}
\begin{proof}
Let $M_1$ and $M_2$ be $N\times N$ matrices and let $D^{\ast}_{M_1}$, $D^{\ast}_{M_2}$, and $D^{\ast}_{cM_1 + M_2}$ be the degree matrices of $M_1$, $M_2$ and $cM_1 + M_2$, respectively, where $c$ is a constant. One obtains
\begin{equation}
D^{\ast}_{cM_1 + M_2} = cD^{\ast}_{M_1} + D^{\ast}_{M_2}\,.
\label{eq:degree_matrix}
\end{equation}
Using Eq.~\eqref{eq:degree_matrix}, we obtain
\begin{align}
L_{cM_1 + M_2} &= D^{\ast}_{cM_1 + M_2} - (cM_1 + M_2)\notag\\
&= cD^{\ast}_{M_1} + D^{\ast}_{M_2}- (cM_1 + M_2)\notag\\
&= c(D^{\ast}_{M_1} - M_1) + (D^{\ast}_{M_2} - M_2)\notag\\
&= cL_{M_1} + L_{M_2}\,,
\label{eq:linearity_laplacian}
\end{align}
which proves the linearity of the map.
\end{proof}
By applying Lemma~\ref{lemma_bound} to the tie-decay adjacency matrix $\displaystyle B(t) =  \sum_{m=0}^M \tilde{c}_m(t)\tilde{A}_m$, we express $L(t)$ as follows:
\begin{equation}
L(t) = \tilde{c}_1(t)L_{\tilde{A}_1} + \cdots + \tilde{c}_M(t)L_{\tilde{A}_M}\,,\label{eq:linearity_tie_decay_laplacian}
\end{equation}
where $L_{\tilde{A}_\ell}$ is the Laplacian matrix of $\tilde{A}_\ell$.
In the next theorem, we use the same notations as those defined in section~\ref{dist_measure}.
\begin{theorem}
If we use the Laplacian distance for the MDS, the embedding coordinate for any $t \in [0, \infty)$ is bounded as follows:
\begin{equation}
\left|\psi(B(t))\right|_2 <  \frac{1}{2}\left[\sqrt{M}\left(\frac{2}{1 - e^{-\alpha\tau}}\right)^2 \left(\sum_{m=1}^M\left\|\tilde{A}_m\right\|_1\right)^2 \sqrt{\sum_{\ell = 1}^k \frac{1}{\lambda_\ell}} + \left| L'_k\vec{\delta}_\mu\right|_2\right],
\label{eq:B(t)-boundedness-Laplacian-final}
\end{equation}
where $\left\|A\right\|_1 = \displaystyle\sum_{i=1}^N \sum_{j=1}^N \left| a_{ij} \right|$ for any $N\times N$ matrix $A = [a_{ij}]$.
\end{theorem}
\begin{proof}
We obtain
\begin{align}
    d_L(L(t),L(t_\ell)) &= \sqrt{\sum_{j=1}^N\left[ \rho_j(L(t)) - \rho_j (L(t_\ell)) \right]^2}\notag\\
    &\le \sum_{j=1}^N\left|\rho_j(L(t)) - \rho_j (L(t_\ell)) \right|\notag\\
    &\leq \sum_{j=1}^N\left|\rho_j(L(t)) \right| + \sum_{j=1}^N\left| \rho_j(L(t_\ell)) \right|.
\label{eq:bounded-L1}
\end{align}
Because any eigenvalue of a Laplacian matrix is non-negative, we obtain
\begin{align}
\displaystyle\sum_{j=1}^N\left|\rho_j(L(t)) \right| &= \sum_{j=1}^N \rho_j(L(t)) \notag\\
    &= \text{tr}(L(t))\notag\\
    &= \tilde{c}_1(t)\text{tr}\left(L_{\tilde{A}_1}\right) + \cdots + \tilde{c}_M(t)\text{tr}\left(L_{\tilde{A}_M}\right).
\label{eq:bounded-L2}
\end{align}
Similarly, we obtain
\begin{equation}
    \sum_{j=1}^N \left| \rho_j(L(t_\ell)) \right| = \tilde{c}_1(t_\ell)\text{tr}\left(L_{\tilde{A}_1}\right) + \cdots + \tilde{c}_M(t_\ell)\text{tr}\left(L_{\tilde{A}_M}\right).
\label{eq:bounded-L3}
\end{equation}
By combining Eqs.~\eqref{eq:bounded-L1}, \eqref{eq:bounded-L2}, \eqref{eq:bounded-L3}, and $\text{tr}(L_{\tilde{A_\ell}}) = \displaystyle\sum_{i=1}^N\sum_{j=1}^N a^{(\ell)}_{ij}= \left\|\tilde{A_\ell}\right\|_1$, we obtain
\begin{align}
    d_L(L(t),L(t_\ell)) &\leq \sum_{m=1}^{M} \left[ \tilde{c}_m(t) + \tilde{c}_m(t_\ell)\right]\left\|\tilde{A}_m\right\|_1\notag\\
    & < \frac{2}{1 - e^{-\alpha\tau}}\sum_{m=1}^{M} \left\|\tilde{A}_m\right\|_1.
\end{align}
Similarly to the derivation of Eq.~\eqref{eq:L'delta_bound}, we obtain
\begin{align}
\left|L'_k\vec{\delta}_B(t)\right|_2^2 &\le \sum_{\ell=1}^k \frac{1}{\lambda_{\ell}}\sum_{\ell' = 1}^M d_L^4(B(t),B(t_{\ell'}))\notag\\
&< M\left(\frac{2 \sum_{m=1}^M\left\|\tilde{A}_m\right\|_1}{1 - e^{-\alpha\tau}}\right)^4 \sum_{\ell = 1}^k \frac{1}{\lambda_\ell}\,.
\label{eq:bound_laplacian_coord_1}
\end{align}
By substituting Eq.~\eqref{eq:bound_laplacian_coord_1} in Eq.~\eqref{eq:bound_frob_coord_1}, we obtain Eq.~\eqref{eq:B(t)-boundedness-Laplacian-final}.
\end{proof}

\section{Trajectories of empirical contact networks}\label{numerical_simulations}

In this section, we apply the proposed temporal network embedding method to two empirical data sets of time-stamped proximity events between human individuals. Both data sets are provided by the Sociopatterns project. The recording was made every 20 seconds. We use the Laplacian distance measure for computing the distance between networks. We set $\alpha=10^{-2}$. The results of the two data sets with $\alpha=10^{-3}$, shown in
Appendix~\ref{appendix_001}, are similar to those with $\alpha = 10^{-2}$, supporting the robustness of our method with respect to the $\alpha$ value.

\subsection{Primary school}\label{Primary_students_contact}

The Primary School data set is a temporal contact network recorded from students and teachers in a primary school in Lyon, France over two days, i.e., 10/1/2009 and 10/2/2009 \cite{gemmetto2014mitigation, stehle2011high}. The data spans from 8:45 AM to 5:20 PM on the first day and from 8:30 AM to 5:05 PM on the second day. We consider the time-stamped contact events between pairs of 242 individuals; there are 8,317 edges and 125,773 contacts. For each day, there are three breaks that allow students to interact with peers in other classes, which are the morning break for 20--25 minutes starting at 10:30 AM, the lunch break between 12 PM and 2 PM, and the afternoon break for 20--25 minutes starting at 3:30 PM.

We show in Figs.~\ref{fig:Primary_sch_trajectory_01}(a) and \ref{fig:Primary_sch_trajectory_01}(b) the trajectories in the two-dimensional embedding space for the first and the second day, respectively. We set $\alpha = 10^{-2}$. In each day, the tie-decay network is initially empty, which corresponds to $(x, y) \approx (-51.42, 14.93)$ in the embedding space. In both days, the trajectory resides in the first quadrant during the morning break (see the points labeled 10:30 AM in Fig.~\ref{fig:Primary_sch_trajectory_01}) and in the fourth quadrant during the lunch break (i.e., 12 PM to 2 PM). The segments of the trajectories during these time windows are relatively, while not completely, distinguishable from those for the rest of the day.

\begin{figure}[t]
    \centering
    \includegraphics[width=\textwidth]{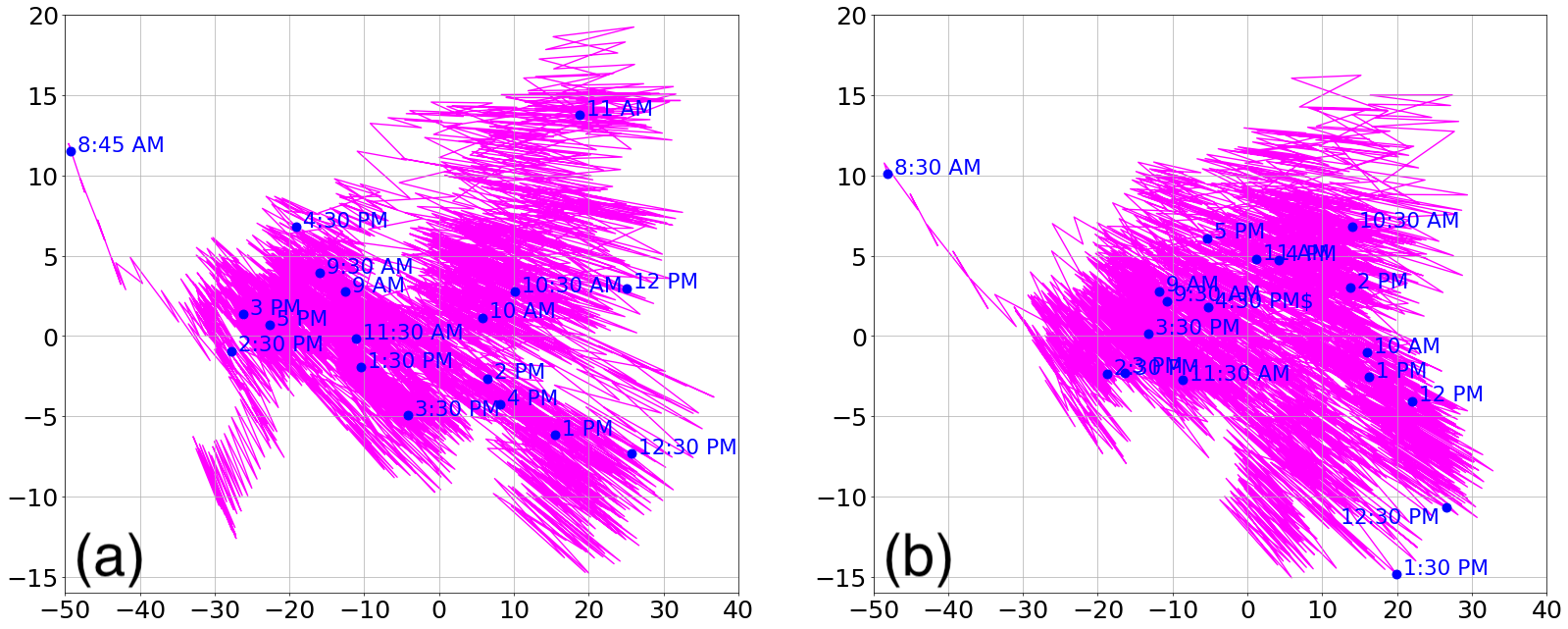}
    \caption{Trajectories of the Primary School temporal networks over two days in the two-dimensional embedding space. (a) Day 1. (b) Day 2. We set $\alpha = 10^{-2}$. The dots indicate the location of the trajectory every 30 minutes.}
    \label{fig:Primary_sch_trajectory_01}
\end{figure}

To compare between the two days, we calculate the average coordinate of the trajectory in the embedding space in each of the consecutive 30-minute windows. We show the average coordinate for the two days in Fig.~\ref{fig:Primary_sch_avg_01}(a). The point labeled 10 AM, for example, represents the $(x, y)$ coordinate averaged between 10 AM and 10:30 AM. Figure~\ref{fig:Primary_sch_avg_01}(a) suggests that the two networks at the same point of time on the two days tend to be located close to each other in the embedding space. For example, the networks at 10:30 AM, corresponding to the morning break, on the two days are both located at $(x, y)$ with large $x$ and $y$ values. As another example, the networks during lunch breaks (i.e., 12 PM to 2 PM) tend to linger around $(x, y)$ with $-10 \le x \le 20$ and $-3 \le y \le 0$.
To quantitatively examine the day-to-day consistency of the network trajectory, we compute the distance between the 30-minute average of the network at the same time between the two days, which we show in Fig.~\ref{fig:Primary_sch_avg_01}(b). First, we observe that the networks during the lunch and afternoon breaks are in particular similar between the two days compared to other times of the daytimes. Second, the distance between the two networks at the same time of the day tends to decrease over the course of the day. Third, the coordinate at the same time of the day is in fact close between the two days. Specifically, the average distance between an arbitrarily chosen pair of times indicated by the dots in Fig.~\ref{fig:Primary_sch_avg_01}(a) is equal to
$16.87$, which is substantially larger than the distance between the two networks at the same time of the day shown in Fig.~\ref{fig:Primary_sch_avg_01}(b), except at 8:30 AM.

\begin{figure}[t]
    \centering
    \includegraphics[width=\textwidth]{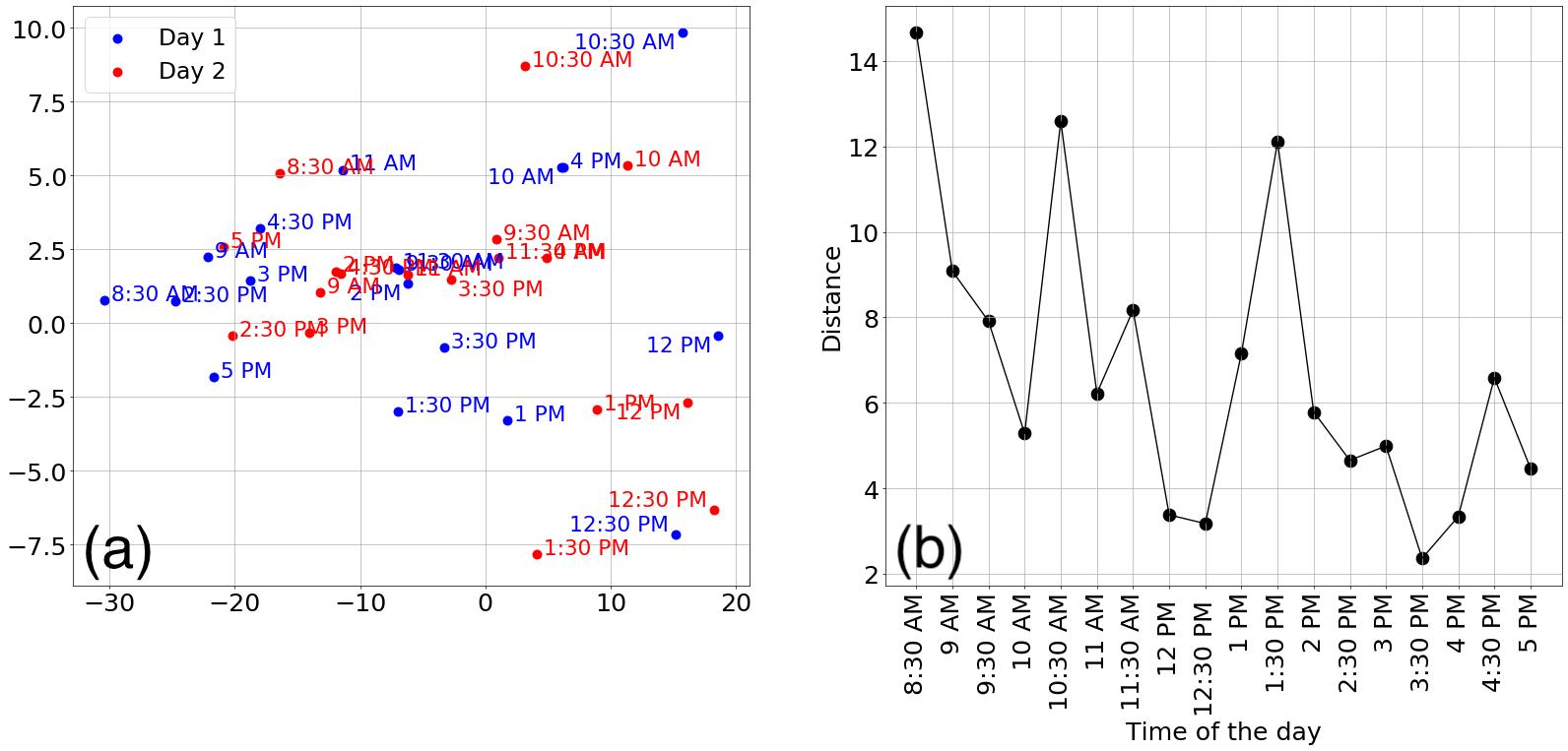}
    \caption{Comparison of the trajectory of the Primary School temporal network between the two days. (a) Average coordinate of the tie-decay network in the embedding space over every 30-minute window. (b) Day-to-day distance between the networks at the same time of the day. We set $\alpha = 10^{-2}$.}
\label{fig:Primary_sch_avg_01}
\end{figure}

To justify our choice of the dimension of the embedding space (i.e., $k = 2$) and the distance measure (i.e., Laplacian distance measure), we carry out the following analysis. We show in Table~\ref{table:goodness_of_fit_01} the values of $g$ and $\sigma$ for $k =1$, $2$, and $3$ for the two distance measures. We find that the Laplacian distance measure yields much larger $g$ values and much smaller $\sigma$ values than 
the Frobenius distance measure does. This result justifies the use of the Laplacian distance. Moreover,  Table~\ref{table:goodness_of_fit_01} shows that, with the Laplacian distance, the embedding dimension of $k=2$ is enough, but not $k=1$.
Note that the $\sigma$ value with $k=2$ is equal to $0.06$, which is smaller than a suggested criterion value of 0.15 \cite{borg2005modern}. With $k=3$, the goodness of embedding further improves compared with the case of $k=2$. However, the gain with $k=3$ is not large because the embedding is already good with $k=2$. Therefore, we conclude that the combination of the Laplacian distance measure and $k=2$ is a reasonable choice for this data set.

\subsection{Hospital}\label{hospital_ward}

The second data set is time-stamped contact network collected from a hospital ward in Lyon, France~\cite{vanhems2013estimating}. We refer to this data set as Hospital. The network is composed of 75 nodes, of which 11 nodes represent medical doctors, 35 nurses, and 29 patients. The network has $1,139$ edges and $32,424$ contacts. The data span the weekdays of a week, from 1 PM on Monday, 12/6/2010 to 2 PM on Friday, 12/10/2010.

To compare the network trajectory across days, we show in Fig.~\ref{fig:Hosp_avg_01}(a) the hourly average of the embedding coordinate. For example, the coordinate labeled 10 AM represents the coordinate of the trajectory averaged between 10 AM and 11 AM. Figure~\ref{fig:Hosp_avg_01}(a) apparently indicates that, similar to the case of the Primary School data, the network trajectory for the different days tends to visit a similar location at the same time of the day. To quantitatively investigate this property, we first compute the distance between the coordinate at a certain time and day and the coordinate at the same time of a different day. Then, given the time of the day, we average the obtained distance value over all possible pairs of the day of the week. For example, There are four tie-decay networks at 10 AM, corresponding to Tuesday, Wednesday, Thursday, and Friday. Therefore, we average the pairwise distance over the 6 possible pairs of days. We show the average distance at each time of the day in Fig.~\ref{fig:Hosp_avg_01}(b). The figure suggests that the tie-decay networks at the same time of the day are apparently close to each other across the days in the morning, but less so in the afternoon, which is opposite to the case of the Primary School data (see Fig.~\ref{fig:Primary_sch_avg_01}(b)). We note that all the day-to-day distance values shown in Fig.~\ref{fig:Hosp_avg_01}(b) are substantially less than 7.61, which is the average distance between an arbitrary pair of hours shown in Fig.~\ref{fig:Hosp_avg_01}(a). Therefore, we conclude that the tie-decay network in different days recurs to a similar network at the same time of the day, similar to the case of the Primary School data.
\begin{figure}[t]
    \centering
        \includegraphics[width=\textwidth]{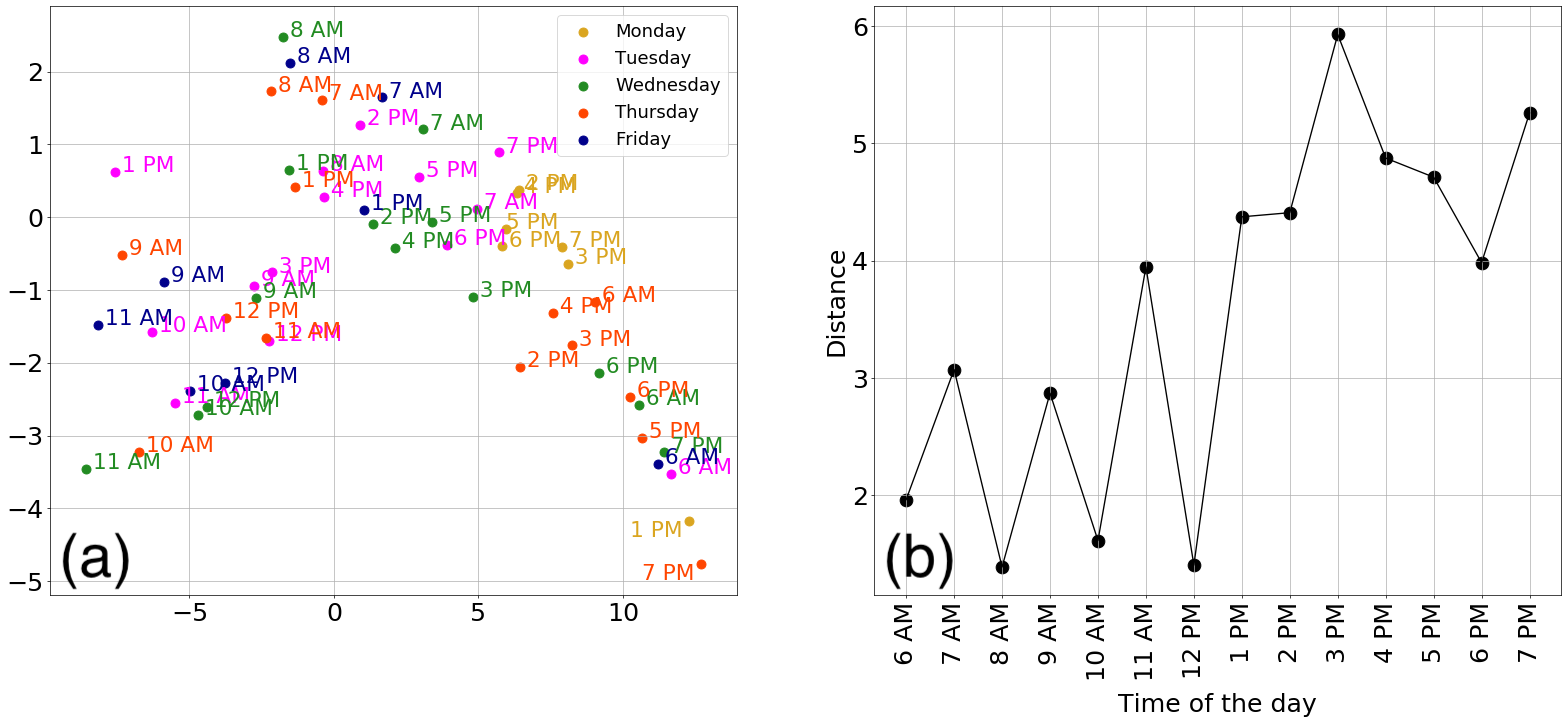}
\caption{Comparison of the trajectory of the Hospital temporal network across the five days. (a) Average coordinate of the tie-decay network in the embedding space over every hourly time window. (b) Day-to-day distance between the networks at the same time of the day, which is averaged over all pairs of days. We omitted the tie-decay networks before 6 AM and after 8 PM because most of the contacts (i.e., 95.6\%) occur between 6 AM and 8 PM. We set $\alpha = 10^{-2}$.}
\label{fig:Hosp_avg_01}
\end{figure}

The previous study showed that contacts among nurses and patients are more abundant in the morning than the afternoon and that those among medical doctors are similar in frequency between the morning and the afternoon~\cite{vanhems2013estimating}.
We verify this observation in Fig.~\ref{fig:Hosp_avg_comparison}(a), which shows the normalized number of contacts of different types in each hour. For normalization, we first divided the number of contacts of each type in each hourly window by the number of the days available. For example, data between 8 AM and 9AM was recorded from Tuesday to Friday. Therefore, we divided the number of contacts between 8 AM and 9 AM by four. Second,
we further normalized the obtained count so that the average over the hours of the day is equal to 1, separately for each type of contact. Similar to
the previous study~\cite{vanhems2013estimating}, we consider three types of contacts, i.e., contacts between medical doctors, labeled MD-MD, those between a medical doctor and a nondoctor, i.e., nurse or patient, which we label MD-nonMD, and those between individuals that are not medical doctors, labeled nonMD-nonMD. Figure~\ref{fig:Hosp_avg_comparison}(a) indicates that the number of MD-MD contacts is relatively stable over the course of the day. In contrast, the MD-nonMD and nonMD-nonMD contacts are roughly more abundant in the morning than the afternoon.

Because the entire network trajectory is more diverse and the MD-MD contacts are relatively more frequent in the afternoon than in the morning, we hypothesize that the temporal subnetwork composed only of the MD-MD contacts varies more considerably across the days of the week than the MD-nonMD and nonMD-nonMD subnetworks. To examine this hypothesis, we analyze the trajectory of the MD-MD, MD-nonMD, and nonMD-nonMD temporal subnetworks. The union of the sets of contacts from the three subnetworks is the set of all contacts in the original data set.
For each subnetwork, we embed the corresponding tie-decay network and compute its hourly average coordinate in the embedding space. Then, for each subnetwork, similar to the calculation for Fig.~\ref{fig:Hosp_avg_01}(b), we compute the average distance between pairs of coordinates at the same time of the day. Finally, we divide this number by the average distance over all the possible pairs of hours. 
If this normalized distance is smaller than $1$, then the trajectory of the subnetwork tends to visit similar locations at the same time of the day across the days.

\begin{figure}[t]
    \centering
    \includegraphics[width=\textwidth]{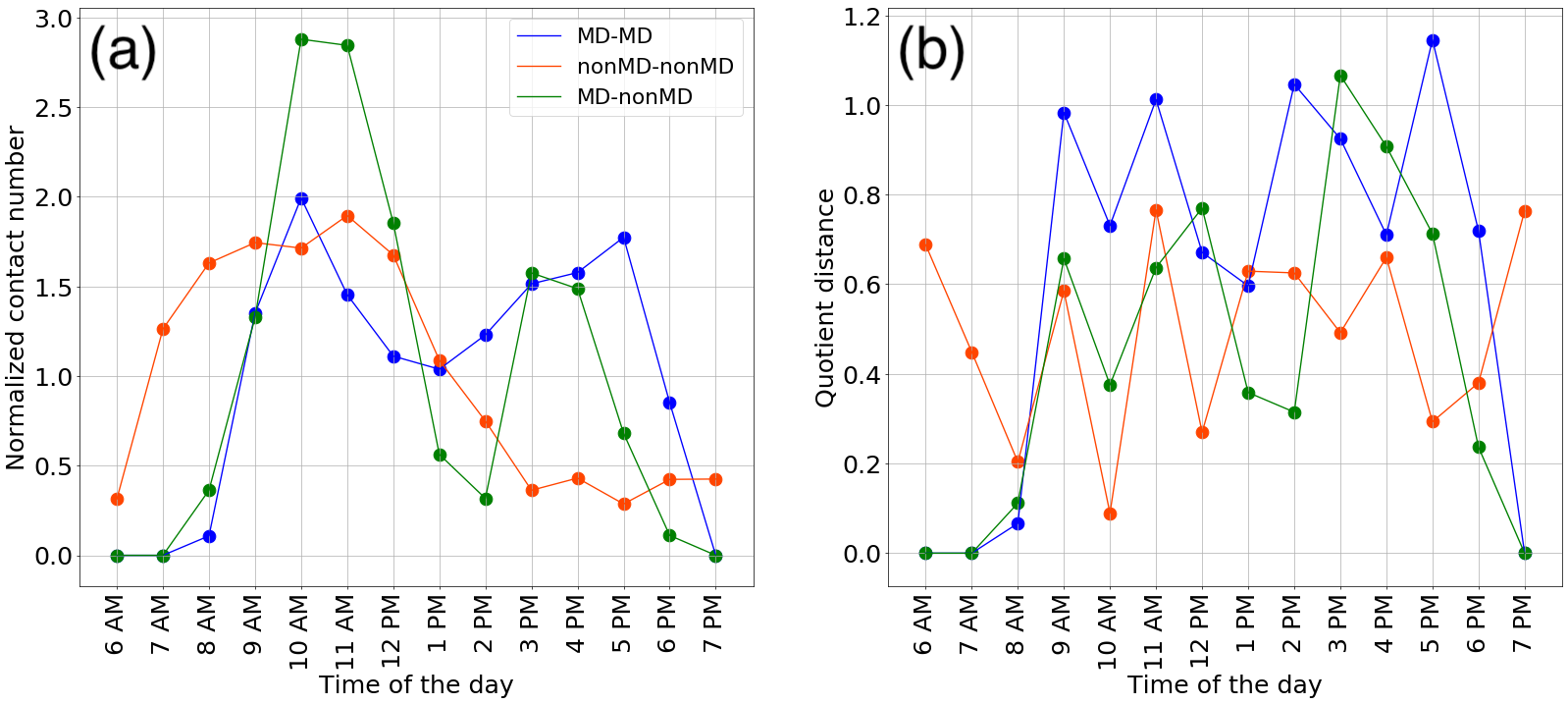}
  \caption{Comparison among temporal subnetworks in the Hospital data. (a) Normalized number of contacts over time. (b) Normalized day-to-day distance at the same time of the day for the different subnetworks. In both (a) and (b), we separately analyzed the trajectories of MD-MD, MD-nonMD, and nonMD-nonMD subnetworks.
}
\label{fig:Hosp_avg_comparison}
\end{figure}

We show the normalized distance for the three subnetworks from 6 AM to 8 PM in Fig.~\ref{fig:Hosp_avg_comparison}(b).
We find that the day-to-day difference of the MD-MD subnetwork at the same time of the day is not particularly small, with the normalized day-to-day distance value being close to 1 across the different hours of the day (shown in blue in Fig.~\ref{fig:Hosp_avg_comparison}(b)).
The day-to-day difference of the MD-nonMD and nonMD-nonMD subnetworks at the same time of the day is smaller than that of the MD-MD subnetwork in most times of the day. In particular, the normalized day-to-day distance for the nonMD-nonMD subnetwork is subsantially smaller than 1 throughout the day. 

Figure~\ref{fig:Hosp_avg_comparison}(a) implies that the fraction of the MD-MD contacts among all contacts is higher in the afternoon than the morning.
Figure~\ref{fig:Hosp_avg_comparison}(b) indicates that the medical doctors form more dissimilar networks across the days than nurses and patients. These two phenomena in combination explain why the day-to-day variation in the entire tie-decay network at the same time of the day is larger in the afternoon than the morning.

We end this section by justifying the choice of the distance measure (i.e., Laplacian distance) and embedding dimension (i.e., $k = 2$). We show
the $g$ and $\sigma$ values for the two distance measures and three values of $k$ in Table~\ref{table:goodness_of_fit_01}. Similar to the case of the Primary School data, we find that the Laplacian distance is much superior to the Frobenius distance,
that $k=2$ is necessary, and that the accuracy of the embedding does not improve much with $k=3$ compared to $k=2$.
The combination of the Laplacian distance and $k=2$ is also satisfactory for embedding the three subnetworks, as we show in  Table~\ref{table:goodness_of_fit_subnetwork_01}.

\begin{table}[t]
\begin{center}
\captionof{table}{Goodness of temporal network embedding. We set $\alpha = 10^{-2}$.\label{table:goodness_of_fit_01}}
\begin{tabular}{|c|c|c|c|c|c|}
 \hline
Data & Distance measure & Index & $k = 1$ & $k = 2$ & $k = 3$\\
\hline
\multirow{4}{*}{Primary School}
& \multirow{2}{*}{Frobenius} & $g$  & 7.73\%  & 13.70\%  & 17.39\% \\
\cline{3-6}
&   & $\sigma$  & 0.81  & 0.76 & 0.73 \\
\cline{2-6}
& \multirow{2}{*}{Laplacian}& $g$ & 84.76\%  & 95.16\%  & 97.13\% \\
\cline{3-6}
& & $\sigma$ & 0.18  & 0.06 & 0.03 \\
\cline{2-6}
 \hline
 \multirow{4}{*}{ Hospital}
 & \multirow{2}{*}{Frobenius} & $g$ & 6.03\%  & 10.51\%  & 14.52\%\\
 \cline{3-6}
& & $\sigma$ & 0.84 & 0.74 & 0.68 \\
\cline{2-6}
&  \multirow{2}{*}{Laplacian} & $g$ & 78.84\% & 90.69\% & 95.13\% \\
\cline{3-6}
& & $\sigma$ & 0.22 & 0.09  & 0.05 \\
\cline{2-6}
 \hline
\end{tabular}
\end{center}
\end{table}

\begin{table}[t]
\begin{center}
\captionof{table}{Goodness of temporal network embedding for the subnetworks of the Hospital's temporal network. We set $\alpha = 10^{-2}$. We use the Laplacian network distance.\label{table:goodness_of_fit_subnetwork_01}}
\begin{tabular}{|c|c|c|}
 \hline
Subnetwork & Index & $k = 2$ \\
\hline
\multirow{2}{*}{MD-MD} & $g$ & 93.23\%  \\
 \cline{2-3}
& $\sigma$ & 0.07 \\
\cline{1-3}
 \multirow{2}{*}{nonMD-nonMD} & $g$ & 91.04\% \\
\cline{2-3}
& $\sigma$ & 0.09  \\
\cline{1-3}
\multirow{2}{*}{MD-nonMD} & $g$ & 91.99\% \\
 \cline{2-3}
& $\sigma$ & 0.08 \\
\cline{1-3}
 \hline
\end{tabular}
\end{center}
\end{table}

\section{Graph classification}

To demonstrate the utility of our embedding method, we apply the Support Vector Machine (SVM) to the embedding coordinates of the Primary School data to carry out a graph classification task \cite{kazemi2020representation, taheri2019learning}. We define four groups of the time of the school day depending on the activities in the school: morning break, lunch break, afternoon break, and class time. We apply the SVM to the set of the embedding coordinates of tie-decay networks averaged over each 30-minute time window (see Fig.~\ref{fig:Primary_sch_avg_01}) to estimate the label of the group to which each 30-minute network has. There are 36 such averaged coordinates, 18 from each of the two days. Then, we compare the classification results to those obtained by applying the SVM to the set of embedding coordinates calculated with baseline methods, i.e., principal component analysis (PCA), t-distributed stochastic neighbor embedding (t-SNE) \cite{van2008visualizing}, and MDS. In the PCA and t-SNE, we compute an adjacency matrix by aggregating the events in each 30-minute time window. We then vectorize the upper-diagonal part of the adjacency matrices into $N(N-1)/2$-dimensional vectors and apply the PCA or t-SNE, obtaining the embedding coordinates in a two-dimensional vector space. In the MDS, we use the same set of aggregated adjacency matrices to compute the distance between each pair of the aggregated adjacency matrices, on which we run the MDS to generate a two-dimensional representation of the adjacency matrices. 
To learn the SVM, we first randomly choose 80\% of the coordinates in the embedding space as the training set. The remaining 20\% of the coordinates are the test set.
Then, we measure the accuracy of the classifier as the number of the correctly classified 30-minute networks divided by the total number of the 30-minute networks in the test set. After repeating this procedure $10^4$ times, we compute the average and the standard deviation of the accuracy.

We show the accuracy values for the tie-decay networks with $\alpha=10^{-2}$ and $10^{-3}$, PCA, t-SNE, and MDS in Table~\ref{table:classification_accuracy}. The mean accuracy for the different classifiers is between $0.8$ and $0.9$ except for the t-SNE, for which the mean accuracy is substantially lower (i.e., $0.65$) than for the other classifiers. The accuracy for the tie-decay networks with $\alpha=10^{-3}$ (i.e., $0.81$) is also notably lower than that for the tie-decay networks with $\alpha=10^{-2}$, PCA, and MDS. 
These results remain similar when the embedding dimension is three, as we show in Appendix~\ref{appendix_3d_classification}.
Then, for the embedding dimension equal to two, we ran a repeated measures ANOVA to test the null hypothesis that the mean accuracy is not different among the remaining three embedding methods, i.e., tie-decay networks with $\alpha = 10^{-2}$, PCA, and MDS. The critical value of the F-statistics corresponding to the significant level of $p=0.05$ is approximately 1.03. Our test yielded a F-statistics value of 0.99, which is less than 1.03. Therefore, we accept the null hypothesis to conclude that our method based on the tie-decay networks with $\alpha = 10^{-2}$ performs on par with PCA and MDS in this graph classification task.

\begin{table}[t]
\begin{center}
\captionof{table}{Mean and standard deviation of accuracy of the different embedding methods in the graph classification task. We set the embedding dimension to two. We calculated the mean and standard deviation on the basis of $10^4$ iterations of the graph classification task.\label{table:classification_accuracy}}
\begin{tabular}{|c|c|c|}
 \hline
Method & Mean & Standard deviation \\
\hline
\multirow{1}{*}{Tie-decay with $\alpha = 10^{-2}$} & $0.88$ & 0.10 \\
\cline{1-3}
\multirow{1}{*}{Tie-decay with $\alpha = 10^{-3}$} & $0.81$ & 0.12 \\
\cline{1-3}
\multirow{1}{*}{PCA} & $0.85$ & 0.12 \\
\cline{1-3}
\multirow{1}{*}{t-SNE} & $0.65$ & 0.15 \\
\cline{1-3}
\multirow{1}{*}{MDS} & $0.89$ & 0.10 \\
\cline{1-3}
 \hline
\end{tabular}
\end{center}
\end{table}


\section{Discussion}

In the present study, we have developed a method to embed temporal network data into an embedding space as a continuous trajectory. 
To this end, we combined the framework of tie-decay networks and the LMDS. We then proved some mathematical properties of the proposed embedding method and applied the method to two empirical data sets of temporal social contact networks among human individuals. For both data sets, we find that the networks at the same time of the day are relatively similar across the different days compared to those at different times of the day. In addition, for the Primary School data set, the coordinates of the morning break and the lunch break in the embedding space are in particular distinguishable from the rest of the day. For the Hospital data set, we have found that the day-to-day variation of the network at the same time of the day is larger in the afternoon than the morning and that this owes to the networking behavior of the medical doctors rather than the nurses and patients.

The network embedding problem that we have addressed in this study is distinct from the more popular node embedding, which aims at mapping individual nodes into a low-dimensional space such that their pairwise similarity is preserved as much as possible. Examples of node embedding include DeepWalk \cite{perozzi2014deepwalk}, node2vec \cite{grover2016node2vec}, and LINE \cite{tang2015line}. Node embedding methods for temporal networks are also available \cite{liben2007link, lu2011link, hasan2011survey}. See \cite{zhang2018network, goyal2018graph} for reviews of node embedding for static and temporal networks.
With node embedding for temporal networks, each node moves in the embedding space as the input network varies over time. 
In contrast, we have proposed to embed the entire network structure at any given time into a single point in the embedding space. Thus, the sequence of networks observed over time generates a continuous-time trajectory lying in the embedding space.
Node embedding for temporal networks and the embedding of a temporal network into a trajectory in the embedding space have different goals. With the latter, we drastically coarse grain the original temporal network by representing the snapshot network at each time point, $B(t)$, by a single point in the embedding space.
In this manner, we are not focusing on the individual nodes and edges, or their relationships. Rather, we aim at capturing both structural and dynamical features of the evolving networks.
%
%
By construction, the distance between two points on the trajectory reflects how different the networks at two times are. This property is potentially useful for downstream tasks such as identifying macroscopic states and dynamics of the system state of temporal networks \cite{masuda2019detecting, sugishita2021recurrence}, change-point analysis \cite{chandola2009anomaly, akoglu2015graph, peixoto2018change}, anomaly detection \cite{ranshous2015anomaly}, systematic identification of recurrence using recurrence quantification analysis \cite{lopes2021recurrence, sugishita2021recurrence}, and link prediction \cite{liben2007link,trouillon2016complex}.

The LMDS is one among many out-of-sample extension methods for dimension reduction, and it can be regarded as an application of the Nystr\"{o}m approximation of the eigenvalues and eigenvectors of matrices \cite{Platt2005Aistats}. The Nystr\"{o}m approximation is useful for our temporal network embedding framework because, with a suitably chosen network distance measure, we can efficiently approximate the elements of the distance matrix that are unobserved. We specifically exploited the exponentially decaying nature of $B(t)$ in the absence of input to calculate the distance between pairs of networks at arbitrary continuous times in a manner compatible with the Nystr\"{o}m approximation. In fact, the LMDS is not the only Nystr\"{o}m-type out-of-sample extension of the MDS, and MDS is not the only dimension reduction method that allows the Nystr\"{o}m algorithm \cite{Platt2005Aistats}. Therefore, our method should also be generalizable to other Nystr\"{o}m-type out-of-sample extensions. Examples include the landmark Isomap algorithm
 \cite{Bengio2004Nips,Desilva2003Nips} (see \cite{Talwalkar2008IeeeConfComputVisPatRecog} for the Nystr\"{o}m algorithm for the Isomap) and the landmark diffusion map based on the Nystr\"{o}m approximation \cite{Long2019ApplComputHarmAnal}. 

Our temporal network embedding method has arbitrary parameters and allows methodological choices. The tie-decay rate, $\alpha$, is defined by the user. The embedding results were qualitatively the same between $\alpha = 10^{-2}$ and $\alpha = 10^{-3}$ for the two empirical temporal networks used in this study. If one uses a large $\alpha$  value, the tie-decay network decays fast such that the network trajectory would quickly approach and linger near the coordinate to which the null network is mapped. In contrast, if $\alpha$ is small, the tie-decay network at time $t$ reflects the input events that are far in the past. Then, the network trajectory may change too slowly and be insensitive to new input events, potentially compromising the interpretation of the results and performance of downstream tasks. Automatic determination of the $\alpha$ value warrants future work.
The choice of the embedding dimension, $k$, and the network distance measure plays an important role.
To this end, we evaluated two indices to assess the embedding quality, concluding that the combination of $k=2$ and the Laplacian network distance is our recommendation. Generalizing this result to a larger set of goodness of fit indices \cite{chennuru2018measures}, network distance measures \cite{donnat2018tracking, masuda2020guide, hartle2020network, wills2020metrics}, and various data of time-stamped contact events also warrants future work.

Although we focused on contact network data, one can apply the proposed framework and methods of temporal network embedding to various temporal network data. For example, dynamic functional connectivity \cite{hutchison2013dynamic, allen2012tracking} and
chronnectome \cite{calhoun2014chronnectome} aim to construct and interpret temporal networks from time-dependent correlational measurements of brain activity.
Our temporal network embedding method and its possible extensions are expected to give new insights into such functional network data in neuroscience. In particular, the location and dynamics of the network trajectory in the embedding space may give us a compact summary of brain dynamics without the need of explicitly
defining and detecting discrete dynamical states of the brain network or their transient dynamics \cite{baker2014fast, ryali2016temporal, VIDAURRE2018646}.
State-transition dynamics has also been applied to data of cell cycle dynamics \cite{lucas2021inferring}.
It may be interesting to apply our methods to these types of data to characterize periodicity and fluctuations around periodic dynamics in cell cycles. 

\section*{Acknowledgements}
The work of Lorenzo Livi was supported by the Canada Research Chairs Program. The work of Naoki Masuda was supported in part by the Air Force Office of Scientific Research (AFOSR) European Office under Grant FA9550-19-1-7024, in part by the National Science Foundation or NSF under Grant DMS-2052720, in part by the Sumitomo Foundation, in part by JSPS KAKENHI under Grant JP 21H04595, and in part by the Japan Science and Technology Agency (JST) under Grant JPMJMS2021.
We also thank SocioPatterns organization for providing the primary school and hospital ward data sets.


\begin{appendices}
\section{Results for $\alpha = 10^{-3}$}
\label{appendix_001}

\begin{table}[t]
\begin{center}
\captionof{table}{Goodness of temporal network embedding. We set $\alpha = 10^{-3}$.\label{table:goodness_of_fit_001}}
\begin{tabular}{|c|c|c|c|c|c|}
 \hline
Data & Distance measure & Index & $k = 1$ & $k = 2$ & $k = 3$\\
\hline
\multirow{4}{*}{Primary School}
& \multirow{2}{*}{Frobenius} & $g$  & 18.59\%  & 33.48\%  & 40.63\% \\
\cline{3-6}
&   & $\sigma$  & 0.69  & 0.60 & 0.56 \\
\cline{2-6}
& \multirow{2}{*}{Laplacian}& $g$ & 82.95\%  & 96.16\%  & 97.77\% \\
\cline{3-6}
& & $\sigma$ & 0.18  & 0.05 & 0.03 \\
\cline{2-6}
 \hline
 \multirow{4}{*}{ Hospital}
 & \multirow{2}{*}{Frobenius} & $g$ & 11.86\%  & 20.70\%  & 28.16\%\\
 \cline{3-6}
& & $\sigma$ & 0.77 & 0.64 & 0.54 \\
\cline{2-6}
&  \multirow{2}{*}{Laplacian} & $g$ & 81.33\% & 92.79\% & 96.70\% \\
\cline{3-6}
& & $\sigma$ & 0.21 & 0.08  & 0.04 \\
\cline{2-6}
 \hline
\end{tabular}
\end{center}
\end{table}

\begin{table}[t]
\begin{center}
\captionof{table}{Goodness of temporal network embedding for the subnetworks of the Hospital's temporal network with $\alpha = 10^{-3}$. We use the Laplacian network distance.\label{table:goodness_of_fit_subnetwork_001}}
\begin{tabular}{|c|c|c|}
 \hline
Subnetwork & Index & $k = 2$ \\
\hline
 \multirow{2}{*}{MD-MD} & $g$ & 94.45\%  \\
 \cline{2-3}
& $\sigma$ & 0.07 \\
\cline{1-3}
 \multirow{2}{*}{nonMD-nonMD} & $g$ & 93.85\% \\
\cline{2-3}
& $\sigma$ & 0.08  \\
\cline{1-3}
\multirow{2}{*}{MD-nonMD} & $g$ & 91.81\% \\
 \cline{2-3}
& $\sigma$ & 0.09 \\
\cline{1-3}
 \hline
\end{tabular}
\end{center}
\end{table}
To assess the robustness of network trajectories in terms of the $\alpha$ value, we repeated the same numerical experiments with $\alpha = 10^{-3}$. Table~\ref{table:goodness_of_fit_001} shows the values of $g$ and $\sigma$ for the Primary School and Hospital data sets, with the two network distance measures and embedding dimensions $k=1,2$, and $3$. The results are qualitatively the same as those with $\alpha=10^{-2}$. Specifically, the Laplacian distance is by far superior to the Frobenius distance, and $k=2$ is necessary and sufficient in practice.
The values of these indices for the subnetworks of the Hospital temporal network, shown in Table~\ref{table:goodness_of_fit_subnetwork_001}, indicate that
the combination of the Laplacian measure and $k = 2$ is also sufficiently good for these subnetworks.
Therefore, in the rest of this section, we use the Laplacian distance and $k=2$.

We show in Fig.~\ref{fig:Primary_sch_trajectory_001} the trajectory of the tie-decay network constructed from the Primary School data in the embedding space, with each panel representing each day. As expected, the trajectory is more slowly changing than in the case of $\alpha=10^{-2}$ (see Fig.~\ref{fig:Primary_sch_trajectory_01}) because a smaller $\alpha$ value used here (i.e., $\alpha = 10^{-3}$) implies a stronger memory effect.
In both days, the trajectory resides in the third quadrant during the morning break (i.e., 10:30 AM) and, except for 12 PM, in the second quadrant during the lunch break (i.e., from 12.30 PM to 2 PM). In addition, the location of the segments of the trajectory in these time windows is relatively distinguishable from the rest of the day. These results are similar to those with $\alpha = 10^{-2}$ (see Fig.~\ref{fig:Primary_sch_trajectory_01}).
\begin{figure}[t]
    \centering
    \includegraphics[width=\textwidth]{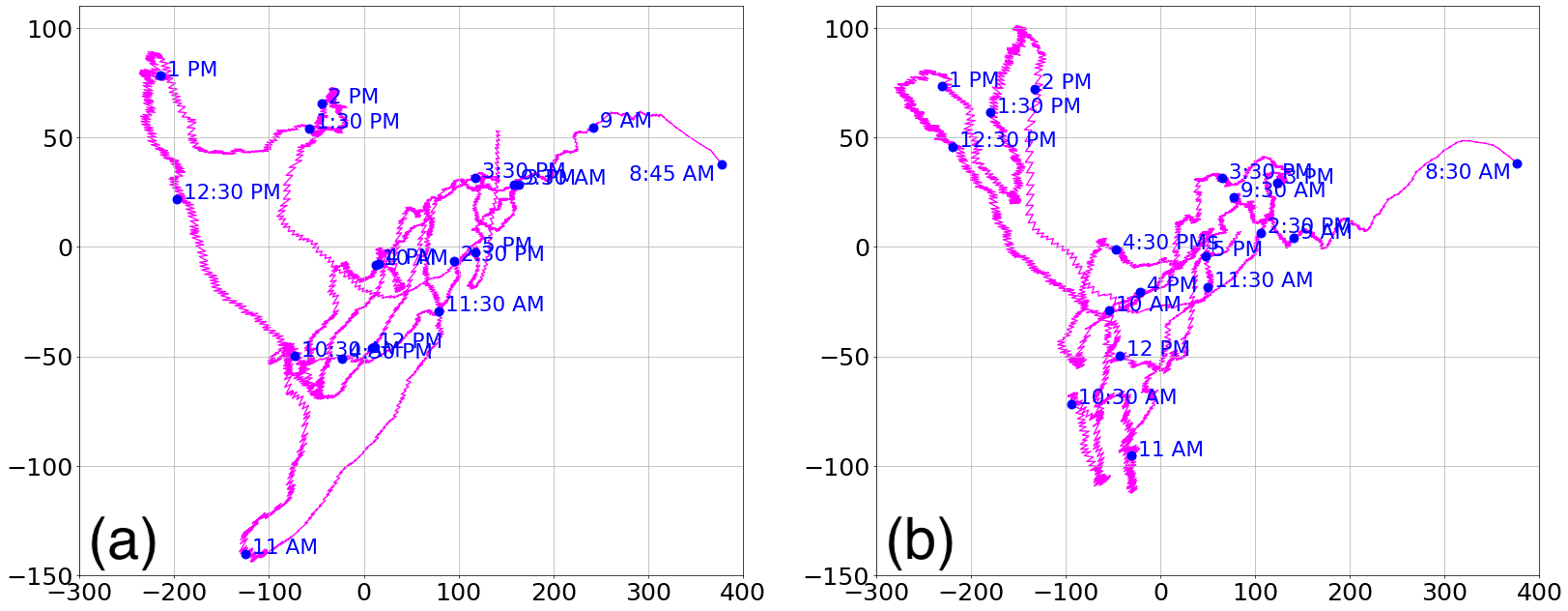}
    \caption{Trajectory of the tie-decay network for the Primary School data. (a) Day 1. (b) Day 2. We set $\alpha = 10^{-3}$. The dots represent the location of the trajectory every 30 minutes.
    }\label{fig:Primary_sch_trajectory_001}
\end{figure}
 We show in Fig.~\ref{fig:Primary_sch_avg_001}(a) the 30-minute average of the embedding coordinate over the two days. The figure suggests that the coordinate at the same time of the day tends to be close between the two days. We show in Fig.~\ref{fig:Primary_sch_avg_001}(b) the distance between the time-averaged tie-decay networks at the same time of the two days in the embedding space. We find that, at any time of the day, the distance is substantially smaller than $156.07$, which is the average distance between all the pairs of 30-minute windows. These results are similar to those with $\alpha=10^{-2}$ (see Fig.~\ref{fig:Primary_sch_avg_01}).
\begin{figure}[t]
    \centering
     \includegraphics[width=\textwidth]{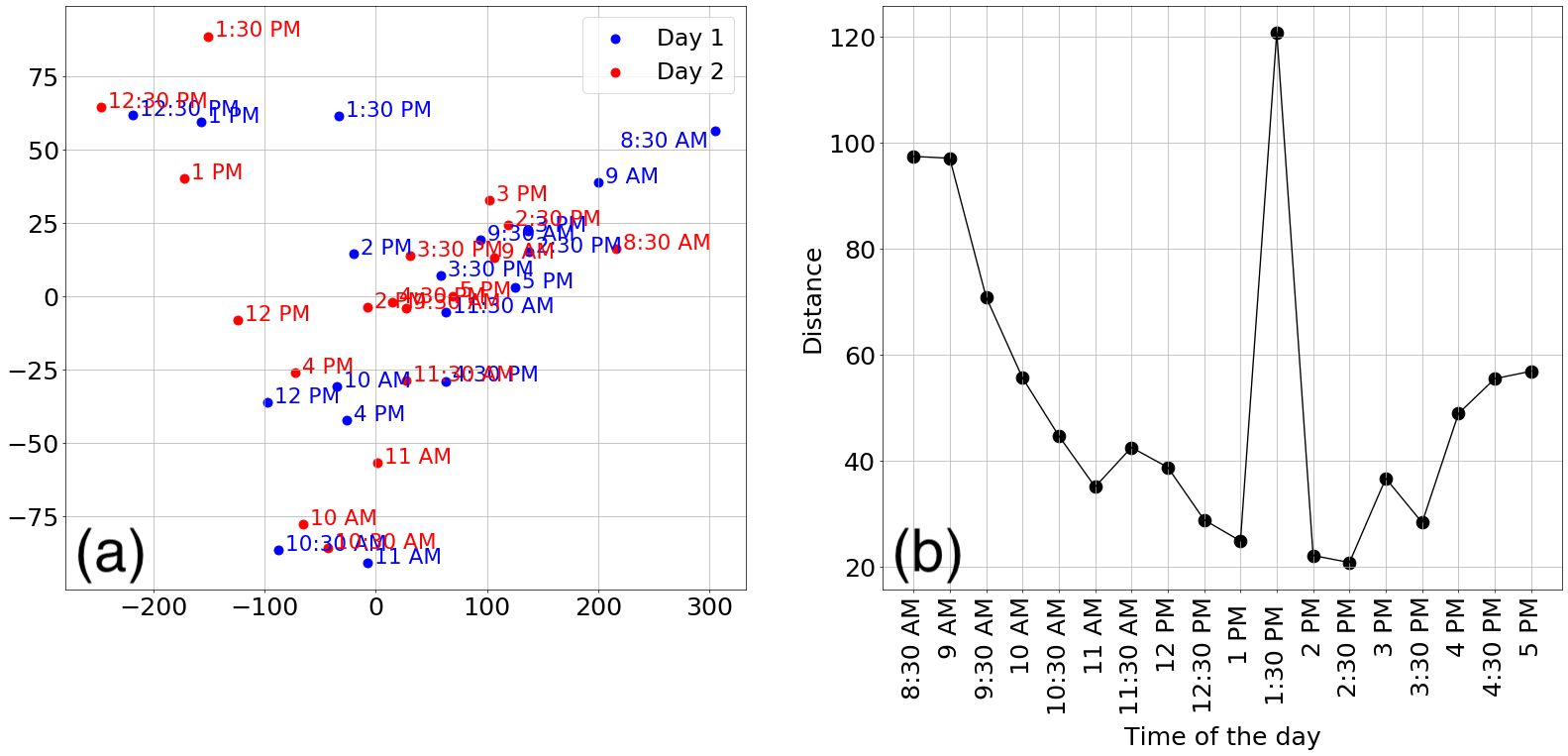}
\caption{Comparison of the trajectory of the Primary School temporal network between the two days, with $\alpha = 10^{-3}$. (a) Average coordinate of the tie-decay network in the embedding space over every 30-minute time window. (b) Day-to-day distance between the networks at the same time of the day.}\label{fig:Primary_sch_avg_001}
\end{figure}

We carried out the same analysis for the Hospital temporal network, with the window size for averaging being an hour.
The results shown in Fig.~\ref{fig:Hosp_avg_001} are similar to those with $\alpha = 10^{-2}$ (see Fig.~\ref{fig:Hosp_avg_01}).
In particular, the average day-to-day distance between the tie-decay networks at the same time of the day, shown in Fig.~\ref{fig:Hosp_avg_001}(b), is substantially smaller than the average distance between an arbitrary pair of networks (i.e., $61.01$), and this tendency is stronger in the morning than in the afternoon. The variability of the day-to-day distance over the course of the day for each subnetwork of the Hospital temporal network is shown in Fig.~\ref{fig:Hosp_avg_comparison_001}. The results are similar to those with $\alpha = 10^{-2}$ (see Fig.~\ref{fig:Hosp_avg_comparison}).

We conclude that our temporal network embedding method yields qualitatively similar trajectories between $\alpha = 10^{-3}$ and $\alpha=10^{-2}$.
\begin{figure}[t]
    \centering
     \includegraphics[width=\textwidth]{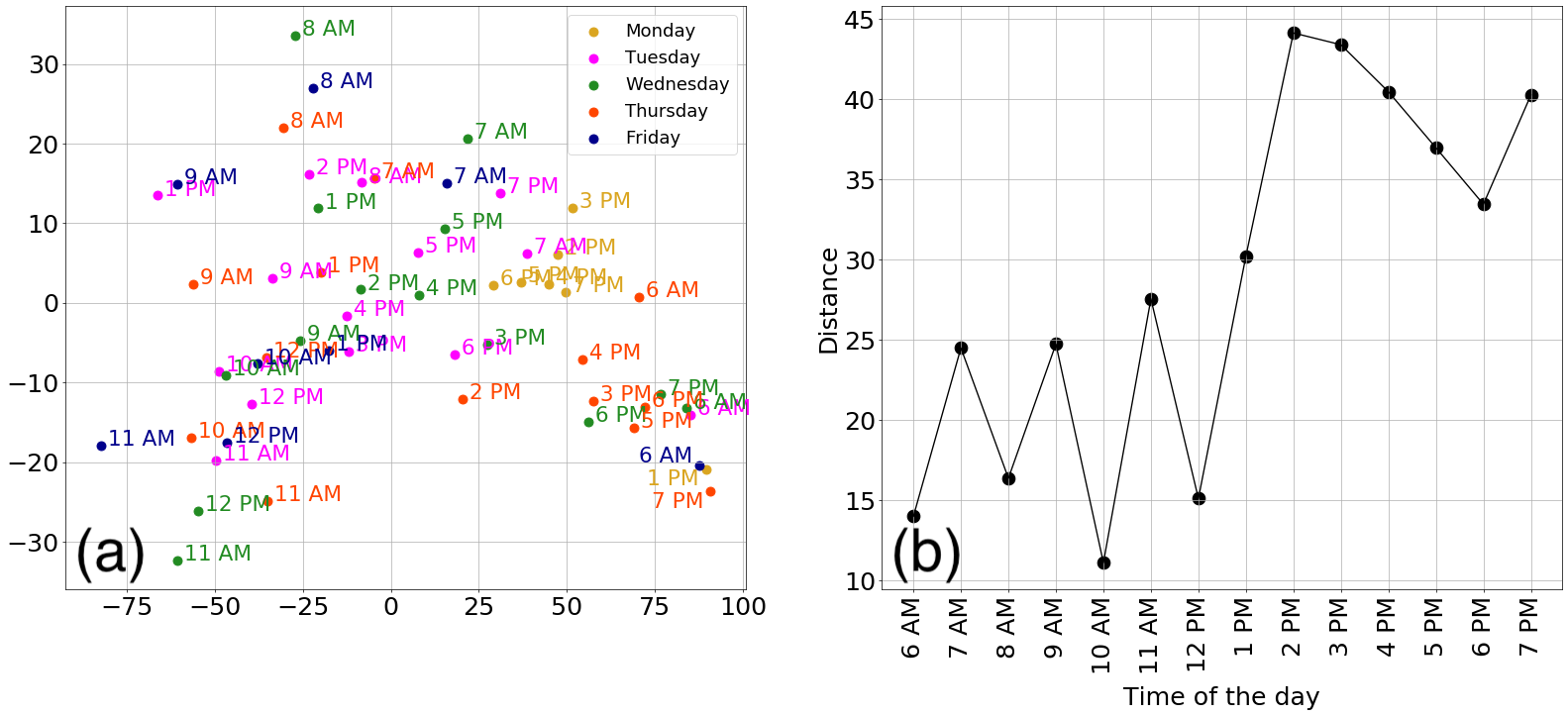}
\caption{Comparison of the trajectory of the Hospital temporal network across the five days, with $\alpha = 10^{-3}$. (a) Average coordinate of the tie-decay network in the embedding space over every hourly time window. (b) Day-to-day distance between the networks at the same time of the day. We omit tie-decay networks before 6 AM and after 8 PM.} 
\label{fig:Hosp_avg_001}
\end{figure}

\begin{figure}[t]
    \centering
    \includegraphics[width=\textwidth]{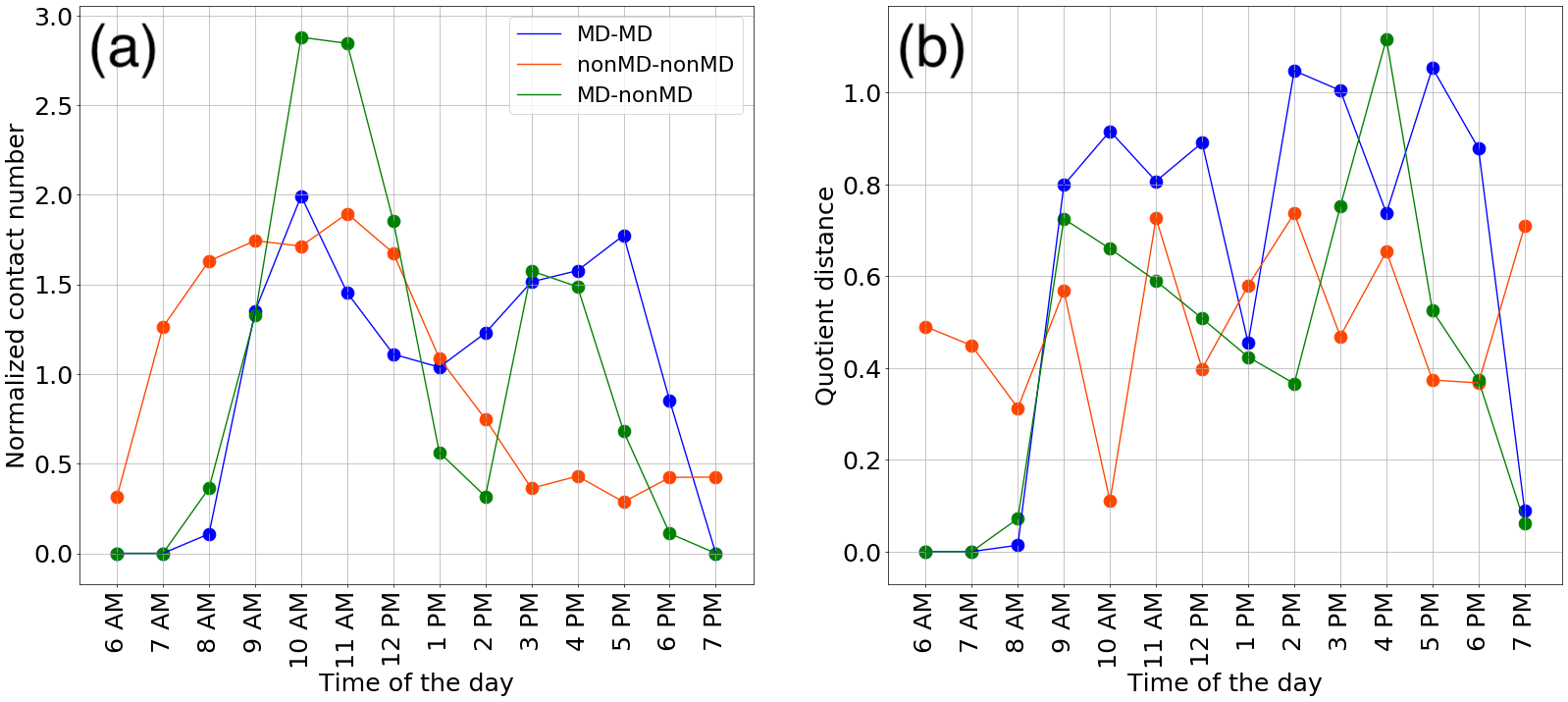}
\caption{Comparison among temporal Hospital subnetworks with $\alpha = 10^{-3}$. (a) Normalized number of contacts over time. (b) Normalized day-to-day distance at the same time of the day for the different subnetworks. In (a), we use the same normalization as that used in Fig.~\ref{fig:Hosp_avg_comparison}(a).}
\label{fig:Hosp_avg_comparison_001}
\end{figure}
\end{appendices}

\begin{appendices}
\section{Results for the graph classification with the three-dimensional embedding space}
\label{appendix_3d_classification}

We compare in Table~\ref{table:classification_accuracy_3dim} the accuracy of the different embedding methods in the graph classification task when the embedding dimension is three.
The results are similar to those with the two-dimensional embedding space, shown in Table~\ref{table:classification_accuracy}.

\begin{table}[t]
\begin{center}
\captionof{table}{Mean and standard deviation of accuracy of the different embedding methods in the graph classification task when the embedding dimension is three. We calculated the mean and standard deviation on the basis of $10^4$ iterations of the graph classification task.\label{table:classification_accuracy_3dim}}
\begin{tabular}{|c|c|c|}
 \hline
Method & Mean & Standard deviation \\
\hline
\multirow{1}{*}{Tie-decay with $\alpha = 10^{-2}$} & $0.88$ & 0.10 \\
\cline{1-3}
\multirow{1}{*}{Tie-decay with $\alpha = 10^{-3}$} & $0.77$ & 0.14 \\
\cline{1-3}
\multirow{1}{*}{PCA} & $0.84$ & 0.13 \\
\cline{1-3}
\multirow{1}{*}{t-SNE} & $0.67$ & 0.15 \\
\cline{1-3}
\multirow{1}{*}{MDS} & $0.86$ & 0.11 \\
\cline{1-3}
 \hline
\end{tabular}
\end{center}
\end{table}

\end{appendices}

\clearpage
\bibliographystyle{chronological}
\bibliography{bibliography.bib}

\end{document}